\documentclass[letterpaper,11pt]{article}

\usepackage{mathptmx} 
\DeclareMathAlphabet\mathcal{OMS}{cmsy}{m}{n}
\SetMathAlphabet\mathcal{bold}{OMS}{cmsy}{b}{n}
\usepackage{titlesec}
\titlespacing*{\section}
{0pt}{2ex}{1ex}
\titlespacing*{\subsection}
{0pt}{1.8ex}{1ex}

\usepackage[utf8]{inputenc}
\usepackage[letterpaper,hmargin=1in,vmargin=1in,footskip=36pt]{geometry}
\usepackage[font=footnotesize,labelfont=bf]{caption}
\usepackage[font=scriptsize,labelfont=bf]{subcaption}

\usepackage{amsthm}
\usepackage{amsmath}
\usepackage{amssymb}
\usepackage{commath}
\usepackage{mathtools}
\usepackage{tikz}
\usetikzlibrary{shapes,patterns,arrows,graphs,decorations.pathreplacing,calc}

\usepackage{graphicx,color}
\usepackage{paralist}
\usepackage{xspace}
\usepackage{hyperref}
\usepackage{booktabs}
\usepackage{multirow}
\usepackage{array}
\usepackage{listings}
\usepackage[inline]{enumitem}
\usepackage{cite} 

\usepackage{thmtools}
\usepackage{thm-restate}

\newtheorem{theorem}{Theorem}
\newtheorem{lemma}{Lemma}

\newtheorem{observation}{Observation}
\newtheorem{corollary}{Corollary}

\newcommand{\eps}{\varepsilon}

\newcommand{\opt}{\textsc{Opt}\xspace}

\newcommand{\pushforward}{\textbf{Push Forward}\xspace}
\newcommand{\pushdown}{\textbf{Push Down}\xspace}
\newcommand{\region}{region algorithm\xspace}
\newcommand{\nextj}{j^*}
\newcommand{\ihat}{i}
\newcommand{\dist}{\varphi}
\newcommand{\I}{\mathcal{I}}
\newcommand{\N}{\mathbb{N}}
\newcommand{\epsfrac}{\frac{\eps}{\eps-\delta}}

\newcommand{\OO}{\mathcal{O}}

\definecolor{pink}{RGB}{255,0,102}

\usepackage[textsize=tiny,textwidth=2cm,color=green]{todonotes}

\usepackage{tcolorbox}
\newtcolorbox{algbox}[1]{colback=white, colframe=black, colbacktitle=white, coltitle=black, width=0.9\textwidth, leftrule=0mm, rightrule=0mm, arc=0mm, title=#1}

\definecolor{root}{RGB}{255,240,0}
\definecolor{level11}{RGB}{0,176,80}
\definecolor{level12}{RGB}{51,204,51}
\definecolor{level13}{RGB}{102,255,51}

\definecolor{level21}{RGB}{0,51,204}
\definecolor{level22}{RGB}{0,76,226}
\definecolor{level23}{RGB}{0,102,255}
\definecolor{level24}{RGB}{0,153,255}

\definecolor{level25}{RGB}{0,204,255}
\definecolor{level26}{RGB}{69,226,255}
\definecolor{level27}{RGB}{139,255,255}

\definecolor{ellipse123}{RGB}{76,229,51}
\definecolor{ellipse2234}{RGB}{0,102,255}
\definecolor{ellipse256}{RGB}{34,215,255}

\definecolor{yellow}{RGB}{255,240,0}
\definecolor{green}{RGB}{51,204,51}
\definecolor{blue}{RGB}{0,102,255}

\tikzstyle{job} = [rectangle,  minimum width = .5cm, minimum height = .5cm, anchor=west, draw=black]
\tikzstyle{whitejob} = [rectangle,  minimum width = .5cm, minimum height = .6cm, anchor=west, draw=white, fill = white]

\title{A general framework for handling commitment\\
in online throughput maximization
}

\author{
Lin Chen\thanks{Department of Computer Science, University of Houston, Texas, United States. Email: \texttt{chenlin198662@gmail.com}.}
\and Franziska Eberle\thanks{Department for Mathematics and Computer Science, University of Bremen, Germany. Email: \texttt{\{feberle,nicole.megow\}@uni-bremen.de}. Supported by the German Science Foundation (DFG) under contract  ME 3825/1.}
 \and Nicole Megow\footnotemark[2] 
 \and Kevin Schewior\thanks{Fakult\"at f\"ur Informatik, Technische Universit\"at M\"unchen, M\"unchen, Germany; D\'epartement d'Informatique, \'Ecole Normale Sup\'erieure, PSL University, Paris, France. Email: \texttt{kschewior@gmail.com}. Supported by Conicyt Grant PII 20150140 and DAAD PRIME program.}
 \and Cliff Stein\thanks{Department of Industrial Engineering and Operations Research, Columbia University, New York, United States. Email: \texttt{cliff@ieor.columbia.edu}. Research supported in part by NSF grants CCF-1421161 and CCF-1714818.}
 }

\date{\today}

\begin{document}
\maketitle

\thispagestyle{empty}

\begin{abstract}
	We study a fundamental online job admission
	problem where jobs with 
	deadlines arrive online over time at their release dates, and the task is to determine a
	preemptive single-server schedule which maximizes the number of jobs that complete on time. To circumvent known impossibility results, we make
	a standard slackness assumption by which the feasible time window
	for scheduling a job is at least $1+\eps$ times its processing time,
	for some $\eps>0$. 
	%
	We quantify the impact that different provider
	commitment requirements have on the performance of online
	algorithms. 
	Our main contribution is
	one universal algorithmic framework for  online job admission both with and
	without commitments. Without commitment, 
	our algorithm with a competitive ratio of~$\OO(1/\eps)$
	is the 
	best possible (deterministic) for this problem. For 
	commitment models, we give the first non-trivial performance bounds.
	If the commitment
	decisions must be made before a job's slack becomes less than a $\delta$-fraction of its size, we prove 
	a competitive ratio of $\OO(\eps/((\eps -\delta)\delta^2))$,
	for $0<\delta<\eps$. When a provider must commit
	upon starting a job, our bound is~$\OO(1/\eps^2)$. 
	Finally, we observe that for scheduling with
	commitment the restriction to the ``unweighted'' throughput model is
	essential; if jobs have individual weights, we rule out 
	\mbox{competitive deterministic~algorithms.}
	%
\end{abstract}



\newpage

\setcounter{page}{1}

\section{Introduction}

Many modern computing environments involve a centralized system for managing the resource allocation for processing many different jobs. Such environments are varied, including, for example, internal clusters and public clouds.  These systems typically handle a diverse workload~\cite{DBLP:conf/spaa/LucierMNY13} with a mixture of jobs including short time-sensitive jobs, longer batch jobs, and everything in between.

The challenge for a system designer is to implement scheduling policies that trade off between these different types of jobs and obtain good performance.  There are many ways to define good performance and in this paper, we will focus on the commonly used notion of {\em throughput} which is the number of jobs completed, or if jobs have weights, the total weight of jobs completed.

In general, throughput 
is a ``social welfare'' objective that tries to maximize total utility. 
By centralizing computing and scheduling decisions, one can potentially better utilize resources. 
Nevertheless, for the ``greater good'' it may be beneficial to abort jobs close to their deadlines in favor of many, but shorter and more urgent tasks~\cite{FergusonBKBF12}. As companies start to outsource mission critical processes to external clouds, they may require a certain provider-side guarantee, i.e., service providers have to {\em commit to complete} admitted jobs before they cannot be moved 
to other computing clusters anymore. Moreover, companies tend to rely on business analytics to support decision making. Analytical tools, that usually 
work with copies of databases, 
depend on faultless 
data. This means, once such a copy process started, its completion must be guaranteed. 


More formally, we consider a fundamental single-machine scheduling model in which jobs arrive online over time at their {\em release date} $r_j$. Each job has a {\em processing time} $p_j\geq 0$, a {\em   deadline}~$d_j$, and possibly a {\em weight} $w_j>0$.  In order to complete, a job must receive a total of $p_j$ units of processing time in the interval $[r_j,d_j]$.  We allow {\em preemption}, that is, the processing time does not need to be contiguous. If a schedule completes a set $S$ of jobs, then the {\em throughput} is $|S|$, while the weighted throughput is $\sum_{j \in S} w_j$. 
	We analyze the performance of algorithms using standard {\em competitive analysis} in which we compare the throughput of an online algorithm with the throughput achievable by an optimal offline algorithm that is given full information in advance. 

Deadline-based objectives are typically much harder to optimize than other Quality-of-Service metrics such as response time or makespan. Indeed, the problem becomes hopeless when {\em preemption} (interrupting a job and resuming it later) is not allowed: whenever an algorithm starts a job $j$ without being able to preempt it, it may miss the deadlines of an arbitrary number of jobs that would have been schedulable if $j$ had not been started. 
For scheduling with commitment, we provide a similarly strong lower bound for the preemptive version of the problem in the presence of weights. Therefore, we focus on {\em unweighted preemptive online} throughput maximization.

Hard examples for online algorithms tend to involve jobs that arrive and then {\em must} immediately be processed since $d_j - r_j \approx p_j$.  It is entirely reasonable to bar such jobs from a system, requiring that any submitted job contains some {\em slack}, that is, we must have some separation between $p_j$ and $d_j - r_j$.  To that end we say that an instance has $\eps$-slack if every job satisfies $d_j - r_j \geq (1 + \eps) p_j$.  We develop algorithms whose competitive ratio depends on $\eps$; the greater the slack, the better we expect the performance of our algorithm to be.  This slackness parameter captures certain aspects of Quality-of-Service (QoS) provisioning and admission control, see e.g. \cite{GeorgiadisGP97,LiebeherrWF96}, and it has  been considered in previous work, \mbox{e.g., in~\cite{DBLP:conf/spaa/LucierMNY13,AzarKLMNY15,SchwiegelshohnS16,GarayNYZ02,Goldwasser1999,BaruahH97}.}
Other results for scheduling with deadlines use speed scaling, which can be viewed as another way to add slack to the 
schedule, e.g. \cite{BansalCP07,PruhsS10,AgrawalLLM18,ImM16}.
In this paper we quantify the impact that different job commitment requirements have on the performance of online algorithms. We parameterize our performance guarantees by the slackness of jobs.

\subsection{Our results and techniques}


Our main contribution is a general algorithmic framework, called \region, for online scheduling with and without commitments. We prove performance guarantees which are either tight or constitute the first non-trivial results. We also answer open questions in previous work.
We show strong lower bounds for the weighted case and therefore 
our algorithms are all for the unweighted case $w_j=1$.

\smallskip
\noindent{\bf Optimal algorithm for scheduling without commitment.} We give an implementation of the \region that achieves a competitive ratio of $\OO(\frac1\eps)$. We prove that this is {\em optimal} by giving a matching lower bound~(ignoring constants) for any deterministic online algorithm.

\smallskip
\noindent{\bf Impossibility results for commitment upon job arrival.} In this most restrictive model an algorithm must decide immediately at a job's release date if the job will be completed or not. 
	We show that no (randomized) online algorithm admits a bounded competitive ratio. 
	Such a lower bound has only been shown by exploiting arbitrary job weights~\cite{DBLP:conf/spaa/LucierMNY13,Yaniv-Thesis2017}. Given our strong negative result, we do not consider this model~any~further. 

\smallskip
\noindent{\bf Scheduling with commitment.} We distinguish two different models: \emph{(i) commitment upon job admission} and \emph{(ii) $\delta$-commitment}. In the first model, an algorithm may discard a job any time before its start, its admission. This reflects the situation when the start of a process is the critical time point after which the successful execution is essential (e.g.\ faultless copy of a database). In the second model, $\delta$-commitment, an online algorithm must commit to complete a job when its slack has reduced from the original slack requirement of an $\epsilon$-fraction of the size to a $\delta$-fraction for~$0 \leq \delta \leq \eps$. 
Then, the latest time for committing to job~$j$ is $d_j - (1 + \delta) p_j$. 
This models an early enough commitment (parameterized by $\delta$) for mission~critical~jobs.

For both models, we show that implementations of the \region allow for the first non-trivial performance guarantees. 
We prove an upper bound on the competitive ratio of~$\OO(1/\eps^2)$ for commitment upon admission and a competitive ratio of $\OO(\eps/((\eps -\delta)\delta^2))$, 
for $0<\delta<\eps$, in the $\delta$-commitment model. These are the first rigorous non-trivial upper bounds in any commitment model~(excluding the special weighted setting with $w_j=p_j$ that has been resolved; see related work).

	
Instances with arbitrary weights are hopeless without further restrictions. There is no deterministic online algorithm with bounded competitive ratio, neither for commitment upon admission (also shown in~\cite{AzarKLMNY15}) nor for $\delta$-commitment. 
	In fact, our construction implies that there is no deterministic online algorithm with bounded competitive ratio in {\em any commitment model} in which a scheduler may have to commit to a job before it has completed. (This 
	is hard to formalize but may give guidance for the design of alternative~commitment~models.) 
	Our lower bound for $\delta$-commitment is actually more fine-grained: 
for any $\delta>0$ and any $\eps$ with $\delta \leq \epsilon < 1+\delta$, no deterministic online algorithm has a bounded competitive ratio for weighted throughput. In particular, this rules out bounded performance guarantees for $\eps\in(0,1)$. We remark that for sufficiently large slackness ($\eps>3$), Azar et al.~\cite{AzarKLMNY15} provide an online algorithm that has bounded competitive ratio. Our new lower bound  answers affirmatively the open question if high slackness is indeed required.
	
Finally, our impossibility result for weighted jobs and the positive result for instances without weights clearly separate the weighted from the unweighted setting. Hence, we do not consider weights in this paper. 

\smallskip
\noindent{\bf Our techniques.}
Once a job $j$ is admitted to the system, its slack becomes a scarce resource: To complete the job before its deadline (which may be mandatory depending on the commitment model, but is at least desirable), one needs 
	to carefully ``spend" the slack on admitting jobs to be processed before 
	the deadline of~$j$. Our general framework for admission control, the region algorithm, addresses this issue by the concept of ``responsibility": Whenever a job $j'$ is admitted while $j$ could be processed, $j'$ becomes responsible for not admitting similar-length jobs for a certain period, its \emph{region}. The intention is that $j'$ reserves time for $j$ to complete. To balance between reservation (
	commitment to complete $j$) and performance (loss of other jobs), the algorithm uses the parameters $\alpha$ and $\beta$, which specify the length of a region and similarity of~job~lengths. 

A major difficulty 
in the analysis of the region algorithm is understanding the complex interval structure formed by feasible time windows, regions, and time intervals during which jobs are processed. Here, we rely on a key design principle of our algorithm: Regions are defined independently of scheduling decisions. Thus, the analysis can be naturally split into two parts:

{\em In the first part,} we argue that the scheduling routine can handle the admitted jobs sufficiently well for aptly 
	chosen parameters $\alpha$ and $\beta$. That means that the respective commitment model is obeyed and, if not implied by that, an adequate number of the admitted jobs is completed.
	
{\em In the second part,} we can disregard how jobs are actually scheduled by the scheduling routine and argue that the region algorithm admits sufficiently many jobs to be competitive with an optimum solution. 
	 The above notion of ``responsibility" 
	 suggests a proof strategy mapping jobs that are completed in the optimum to the corresponding job that was ``responsible" due to 
	 its region. Transforming this idea into a charging scheme is, however, a non-trivial task: 
	 There might be many (
	 $\gg\OO(\frac{1}{\eps^2}$)) jobs released within the region of a single job~$j$ and completed by the optimum, but not admitted by the region algorithm due to many consecutive regions of varying size.
	 It is unclear where to charge these jobs---clearly not all of \mbox{them to $j$.} 
	
	 We develop a careful charging scheme that avoids such overcharging. We handle the complex interval structure by working on a natural tree structure (\emph{interruption tree}) related to the region construction and independent of the actual schedule. 
	  Our charging scheme comprises two central routines for distributing charge: 
	  Moving charge along a sequence of consecutive jobs (\emph{Push Forward}) or to children~(\emph{Push~Down}).
	  
We show that our analysis of the region algorithm is tight up to a constant factor. 


\subsection{Previous results}

Preemptive online scheduling and admission control have been studied rigorously. There are several results regarding the impact of deadlines on online scheduling; see, e.g.,~\cite{Goldwasser1999,GarayNYZ02,DBLP:conf/rtss/BaruahHS94} and 
	references therein. Impossibility results for jobs with hard deadlines and without slack have been known for decades~\cite{DBLP:journals/rts/BaruahKMMRRSW92,DBLP:journals/tcs/KorenS94,DBLP:journals/siamcomp/KorenS95,CanettiI98,Lipton1994}.




\emph{Scheduling without commitment.} Most research on online 
scheduling does not address commitment. The only results independent of slack (or other job-dependent parameters) concern the machine utilization, i.e., weighted throughput for the special case $w_j=p_j$, where a constant competitive ratio is possible~\cite{DBLP:journals/rts/BaruahKMMRRSW92,DBLP:journals/tcs/KorenS94,DBLP:journals/siamcomp/KorenS95,Woeginger1994}. 
	In the unweighted setting, a randomized~$\OO(1)$-competitive algorithm is known~\cite{DBLP:journals/jal/KalyanasundaramP03}.
	%
	%
For instances with $\eps$-slack, Lucier et al.~\cite{DBLP:conf/spaa/LucierMNY13} give an~$\OO(\frac{1}{\eps^2})$-competitive algorithm in the most general weighted setting. To the best of our knowledge, no lower bound was known to date. 
	

\emph{Scheduling with commitment.} Much less is known for scheduling with commitment. In the most restrictive model, \emph{commitment upon job arrival}, Lucier et al.~\cite{DBLP:conf/spaa/LucierMNY13} rule out competitive online algorithms for any slack parameter~$\eps$ 
	when jobs have arbitrary weights. For {\em commitment upon job admission}, they give a heuristic that empirically performs very well but for which they cannot show a rigorous worst-case bound. In fact, later Azar et al.~\cite{AzarKLMNY15} show that no bounded competitive ratio is possible for weighted throughput maximization for small~$\eps$. For the {\em $\delta$-commitment model}, Azar et al.~\cite{AzarKLMNY15} design (in the context of truthful mechanisms) an online algorithm 
	that is $\OO(\frac{1}{\eps^2})$-competitive if the slack~$\eps$ is sufficiently large. They call an algorithm in this model {\em $\beta$-responsive algorithm}. They left open if this latter condition is an inherent property of any committed scheduler in this model and we answer this affirmatively.
	
	Again, the machine utilization variant ($w_j=p_j$) is much more tractable than weighted or unweighted throughput maximization. Simple greedy algorithms achieve the best possible competitive ratio~$\Theta(\frac{1}{\eps})$~\cite{DBLP:conf/approx/DasGuptaP00,GarayNYZ02} in all aforementioned commitment models, even commitment upon arrival.

\section{Our general framework}\label{sec:RegionAlgorithm}

\subsection{The region algorithm}

In this section we present our general algorithmic framework which we apply to scheduling with and without commitment. We assume that the slackness constant $\eps>0$ is known to the online algorithm.
In the $\delta$-commitment model it is sufficient that $0 < \delta < \eps$ is known. 

To gain some intuition for our algorithm, we first describe informally the three underlying design principles. They also explain our algorithm's parameters.
\begin{compactenum}
	\item A running job can be preempted only by significantly smaller jobs (parameter $\beta$).
	\item A job cannot start for the first time when its remaining slack is too small (constant $\delta$ which is part of the input in the $\delta$-commitment model and otherwise set to $\delta=\frac{\eps}{2}$).
	\item If a job preempts other jobs, then it has to take ``responsibility'' for a certain time interval (parameter~$\alpha$) with which it assures that the jobs it preempted can complete on time. 
\end{compactenum}

The first two design principles have been used
before~\cite{DBLP:conf/spaa/LucierMNY13}. However, to improve on
existing results we crucially need the third principle. We implement
it in the following way.


The \region has two parameters, $\alpha \geq 1$ and $0 < \beta
< 1$. 
A {\em region} is a union of time intervals associated with a job, and the size of the region
is the sum of sizes of the intervals.
We denote the region of job $j$ by $R(j)$.    Region $R(j)$ will always have size~$\alpha p_j$,
although the particular time intervals composing the region may change over time.  Regions are always disjoint, i.e., for any $i \neq j$, $R(i) \cap R(j) = \emptyset$.
Informally,  whenever our algorithm starts a job $i$ (we say $i$ is {\em admitted}) that arrives
during the region of an already admitted job $j$, then the current interval of $j$ is split into two intervals and the region~$R(j)$ as well as all later regions are delayed.

Formally speaking, at any time $t$, the region algorithm maintains two sets of jobs: \textit{admitted jobs}, which have been started before or at time $t$, and \textit{available jobs}.  A job $j$ is availabe 
if it is released before or at time~$t$, is not yet admitted, and it is not too close to its deadline, i.e., 
$r_j \leq t$ and  $d_j - t \geq (1 + \delta) p_j$.
The intelligence of the \region lies in admitting jobs and
(re)allocating regions. The actual scheduling decision then is simple
and independent of the regions: at any point in time, schedule the
shortest admitted job that has not completed its processing time,
i.e., we schedule admitted jobs in {\em Shortest Processing Time~(SPT)}
order.  The \region  never explicitly considers deadlines except when deciding whether to admit jobs.

The \region starts by admitting job $1$ at its
release date and creates the region~$R(1) := [r_1,r_1+\alpha p_1)$. 
There are two events that trigger a decision of the \region: the release of a job and the end of a region. If one of these events occurs at time~$t$, the \region invokes the {\bf region preemption} subroutine. 
This routine compares the processing time of the smallest \textit{available} job~$i$ with the processing time of the {\em admitted} job~$k$ whose region contains~$t$.  If $p_i < \beta p_k$, job~$i$ is admitted and the \region reserves the interval $[t,t + \alpha p_i)$
for processing~$i$. 
Since regions must be disjoint, the algorithm then modifies all other remaining regions, i.e., 
the parts of regions that belong to~$[t,\infty)$ of other jobs~$j$. We refer to the set of such jobs~$j$ whose regions have not yet completed by time~$t$ as~$J(t)$. %
Intuitively, we preempt the interval of the region containing $t$ and delay its remaining part as well as the remaining regions of all other jobs. Formally, this \textbf{update of all remaining regions} is defined as follows. Let $k$ be the one job whose region is interrupted at time $t$, and let $[a_k',b_k')$ be the interval of $R(k)$ containing $t$. Interval $[a_k',b_k')$ is  replaced by $[a_k',t)\cup[t+\alpha p_i, b_k'+ \alpha p_i)$. For all other jobs $j \in J(t)\setminus\{k\}$, the remaining region $[a_j', b_j')$ of~$j$ is replaced by $[a_j'+\alpha p_i, b_j'+\alpha p_i )$. 
Observe that, although the region of a job may change throughout the algorithm, the starting point of a region for a job will never be changed. We summarize the \region in Algorithm~\ref{alg:regionalg}.  

\begin{lstlisting}[float, caption={Region algorithm}, label={alg:regionalg},belowskip=-0.5 \baselineskip]
Scheduling routine: At any time $t$, run an admitted and not yet completed job with shortest processing time. @\smallskip@	
Event: Upon release of a new job at time $t$:
	Call region preemption routine. @\smallskip@	
Event: Upon ending of a region at time $t$:
	Call region preemption routine. @\smallskip@ 
Region preemption routine:
$k$ @\assign@ the job whose region contains $t$
$i$ @\assign@ a shortest available job at $t$, i.e., $i = \arg\min\{p_j \,|\, r_j \leq t \text{ and } d_j - t \geq (1 + \delta) p_j\}$
If $p_i < \beta p_k$, then
	1. admit job $i$ and reserve region $R(i)=[t,t+\alpha p_i)$,
	2. update all remaining regions $R(j)$ with $R(j)\cap [t,\infty) \neq \emptyset$ as described below.
\end{lstlisting}

%
%
%

We apply the algorithm in different commitment models with different choices of parameters~$\alpha$ and~$\beta$, which we derive in the following sections. In the $\delta$-commitment model, $\delta$ is given as part of the input. In the other models, i.e., without commitment or with commitment upon admission, we simply set~$\delta=\frac{\eps}{2}$.

\smallskip
\noindent {\bf Commitment.} The \region commits always upon admission, i.e., at its first start. 
This is possibly earlier than required in the $\delta$-commitment model. The parameter $\delta$ determines the latest possible start time of a job, which is then for our algorithm also the latest time the job can be admitted. Thus, for the analysis, the algorithm execution for commitment upon admission (with $\delta=\frac{\eps}{2}$) is a special case of $\delta$-commitment. This is true only for our algorithm, not in general. 


\subsection{Main results on the region algorithm}\label{subsec:Analysis}


In the analysis we focus on instances with small slack as they constitute the hard case. We assume for the remainder that $0<\eps\leq 1$. 
%
Notice that instances with large slack clearly satisfy a small slack assumption. We run our algorithm simply by setting $\eps=1$ and obtain constant competitive~ratios.



Our main results are as follows. Without commitment, we present an optimal online algorithm.

\begin{theorem}[Scheduling without commitment]\label{thm:UB-nocommitment}
	Let $0<\eps \leq 1$. With the choice of $\alpha=1$, $\beta = \frac{\eps}{4}$, and $\delta = \frac{\eps}{2}$, the \region is $\Theta(\frac{1}{\eps})$-competitive for scheduling without commitment. 
\end{theorem}
This is an exponential improvement upon the previously best known upper bound~\cite{DBLP:conf/spaa/LucierMNY13} (given for weighted throughput).
For scheduling with commitment, we give the first rigorous upper bound. 
\begin{theorem}[Scheduling with commitment]\label{thm:UB-commitment}
	Let $0< \delta < \eps \leq 1$. Choosing $\alpha=\frac{8}{\delta}$, $\beta = \frac{\delta}{4}$, the \region is $\OO(\frac{\eps}{(\eps-\delta)\delta^2})$-competitive in the $\delta$-commitment model. When  the scheduler has to commit upon admission, the \region has a competitive ratio $\OO(\frac{1}{\eps^2})$ for $\alpha=\frac{4}{\eps}$ and $\beta = \frac{\eps}{8}$.
\end{theorem}

In Appendix~\ref{sec:TightAnalysis}, we show that the analysis of our framework is tight up to constants.




\subsection{Interruption trees}

To analyze the performance of the region algorithm on a given instance, we consider the final schedule and the final regions 
and investigate them retrospectively. Let~$a_j$ be the admission date of job $j$ which remained fixed throughout the execution of the algorithm. Let~$b_j$ denote the final end point of~$j$'s region.  

%

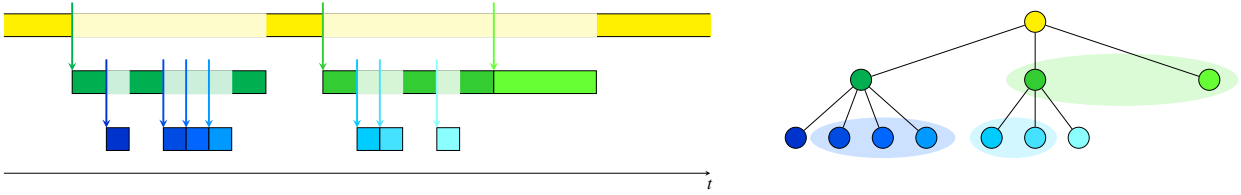
\begin{figure}[tbp]
	\centering
	\begin{subfigure}[t]{.58\textwidth}
		\resizebox{\textwidth}{!}{%
			\begin{tikzpicture} 
			
			\def \r {0}; 
			
			\node[circle, draw=none, fill=none, inner sep=0pt] at (\r, .75) {};
			\node[circle, draw=none, fill=none, inner sep=0pt] at (15.5, -3.5) {};
			
			\def \loo {1.5};			
			\def \lot {7};
			\def \loth {10.75};
			
			\def \lto {2.25};
			\def \ltt {3.5};
			\def \ltth {4};			
			\def \ltf {4.5};
			\def \ltfi {7.75};
			\def \lts {8.25};
			\def \ltse {9.5};
			
			\def \lvldist {1.25}
			\def \jobheight {.5}
			\def \arrowshift {.25}
			
			\draw[-stealth] (0,-3) to (15.5,-3) node [below] {$t$};
			
			\fill[fill=root, draw=none] (\r,0) rectangle (\r + 15.5,\jobheight);
			\draw[draw=black] (\r,0) to (\r + 15.5,0);
			\draw[draw=black] (\r,\jobheight) to (\r + 15.5,\jobheight);
			
			\foreach \job/\pos/\len in {1/\loo/4.25, 2/\lot/3.75, 3/\loth/2.25}{
				\filldraw[fill = level1\job, draw = black, ] (\pos,-\lvldist) rectangle (\pos + \len,-\lvldist + \jobheight);
				\filldraw[draw = root!20, fill = root!20] (\pos,0) rectangle (\pos + \len,0 + \jobheight);
				
				\draw [-stealth, level1\job, very thick] (\pos,\jobheight + \arrowshift ) to (\pos, -\lvldist + \jobheight);
			}
			
			\foreach \job/\pos/\parent in {1/\lto/1, 2/\ltt/1, 3/\ltth/1, 4/\ltf/1, 5/\ltfi/2, 6/\lts/2, 7/\ltse/2}{ 
				\filldraw [fill = level2\job, draw=black]  (\pos,-2*\lvldist) rectangle  (\pos + \jobheight, -2*\lvldist + \jobheight);
				\filldraw[draw = level1\parent!20, fill = level1\parent!20 ] (\pos,-\lvldist)  rectangle (\pos + \jobheight,-\lvldist + \jobheight) ;
				\draw [-stealth, level2\job, very thick] (\pos, \jobheight + \arrowshift - \lvldist) to (\pos, \jobheight - 2*\lvldist);
			}
			\end{tikzpicture}
		}
	\end{subfigure}
	\hfill
	\begin{subfigure}[t]{.38\textwidth}
		\resizebox{\textwidth}{!}{%
			\begin{tikzpicture} [grow=down, level 1/.style={sibling distance=3cm},
			level 2/.style={sibling distance=.75cm}, level distance= 1cm]	
			
			
			\node[circle, draw=none, fill=none, inner sep=0pt] at (-4.5, .5) {};
			\node[circle, draw=none, fill=none, inner sep=0pt] at (3.5, -3) {};
			
			\fill [fill=ellipse123!20] (1.5,-1) ellipse [x radius=2cm,y radius=.45cm];
			\fill [fill=ellipse256!20] (-.375,-2) ellipse [x radius=.75cm,y radius=.35cm];
			\fill [fill=ellipse2234!20] (-2.625,-2) ellipse [x radius=1.25cm,y radius=.35cm];
			
			\node [circle,draw=black,fill=root,minimum width= .1cm] {} 
			child { node [circle,draw=black,fill=level11,minimum width= .1cm] {}
				child { node [circle,draw=black,fill=level21,minimum width= .1cm] {} }
				child { node [circle,draw=black,fill=level22,minimum width= .1cm] {} }
				child { node [circle,draw=black,fill=level23,minimum width= .1cm] {} }
				child { node [circle,draw=black,fill=level24,minimum width= .1cm] {} }
			}
			child { node [circle,draw=black,fill=level12,minimum width= .1cm] {}
				child { node [circle,draw=black,fill=level25,minimum width= .1cm] {} }
				child { node [circle,draw=black,fill=level26,minimum width= .1cm] {} }
				child { node [circle,draw=black,fill=level27,minimum width= .1cm] {} }
			}
			child { node [circle,draw=black,fill=level13,minimum width= .1cm] {} 
			};
			\end{tikzpicture}
		}
	\end{subfigure}
	\caption{Gantt chart of the regions (left) and the interruption tree (right) generated by the region algorithm.}\label{fig:Region:ScheduleAsTree}
\end{figure}

Our analysis crucially relies on understanding the interleaving structure of the regions that the algorithm constructs. This structure is due to the interruption by smaller jobs and can be captured well by a 
	tree or forest. Every job is represented by one vertex. 
A job vertex is the child of another vertex if and only if the region of the latter is interrupted by the first one. The leaves 
correspond to 
jobs with non-interrupted regions. 
By adding a machine job $M$ with $p_M := \infty$ and $a_M=-\infty$, we can assume that the instance is represented by a tree which we call {\em interruption tree}. This idea is visualized in Figure \ref{fig:Region:ScheduleAsTree}, where the vertical arrows indicate the interruption of a region by another job and intervals of the same color belong to one job. Let~$\pi(j)$ denote the \textit{parent} of~$j$. Further, let $T_j$ be the subtree of the interruption tree rooted in job $j$ and let the forest $T_{-j}$ be $T_j$ without its root $j$. By slightly abusing notation, we denote the tree/forest as well as its job vertices by $T_{*}$. 

A key property of this tree is that the processing times on a path 
are geometrically decreasing.




\begin{lemma}
\label{lem:GeometricallyDecreasing}
	Let $j_1,\ldots,j_{\ell}$ be $\ell$ 	jobs on a path in the interruption (sub)tree $T_j$ rooted in $j$ such that $\pi(j_{i+1}) = j_{i}$. Then, $p_{j_\ell} \leq \beta p_{j_{\ell-1}} 
	\cdots \leq \beta^{\ell-1} p_{j_1} \leq \beta^{\ell}p_j$ and the total processing volume is $\sum_{i=1}^{\ell} p_{j_i}\le \sum_{i=1}^{\ell} \beta^i p_{j} \leq \frac{\beta}{1-\beta}\cdot p_{j}$.
\end{lemma}
\begin{proof}
	Let the jobs~$j_1,\ldots,j_{\ell}$ be indexed in decreasing order of processing times, i.e.,~$p_{j_{\ell} }\leq p_{j_{\ell -1}} \leq \ldots \leq p_{j_1}$. 
	
	Observe that~$p_{j_1} < \beta p_j$ as otherwise the region of~$j$ shall not be preempted by~$j_1$. Furthermore, for any~$2\le i\le \ell$, we claim that job~${j_i}$ is released after~$j_{i-1}$. Suppose the claim is not true, then for some~$i$ job~$j_i$ is released before~$j_{i-1}$. Consider the point in time~$t$ when job~$j_{i-1}$ is admitted. The time~$t$ either belongs to the region of~$j_i$, or belongs to the region of some job~$j'$ which interrupts the region of~$j_i$, and consequently~$p_{j'} < \beta p_{j_i}$. In both cases the algorithm will 
	not admit $j_{i-1}$, and therefore the claim is true. At any point in time when the algorithm admits a job~$j_i$, then it interrupts the region of~$j_{i-1}$ and~$p_{j_i} < \beta p_{j_{i-1}}$. Thus, we have
	\[
	p_{j_\ell} < \beta p_{j_{\ell-1}} < \beta^3 p_{j_{\ell-2}} < \cdots < \beta^{\ell-1} p_{j_1} < \beta^{\ell}p_j.
	\]
	The total processing volume of the jobs $j_1,\ldots,j_\ell$ is
	\begin{eqnarray}\label{eq:1}
	\sum_{i=1}^{\ell} p_{j_i} < \sum_{i=1}^{\ell} \beta^{i}p_j 
	= \frac{\beta(1-\beta^{\ell})}{1-\beta}\cdot p_j
	\leq \frac{\beta}{1-\beta}\cdot p_j.
	\end{eqnarray}
\end{proof}

\section{Successfully completing sufficiently many admitted jobs}\label{sec:Feasibility}

We show that the \region completes sufficiently many jobs among the admitted jobs before their deadline. For scheduling without commitment, we show how to choose~$\alpha, \beta$, and~$\delta$ to ensure that {\em at least half} of all admitted jobs are completed on time. 
For scheduling with commitment, we provide a choice of~$\alpha, \beta$, and~$\delta$ such that {\em every} admitted job is guaranteed to complete on time. 

\subsection{Scheduling without commitment}\label{subsec:NoCommitment}

In this section we fix $\delta=\frac{\eps}{2}$ for $0< \eps \leq 1$. We show the following result.
\begin{theorem}\label{thm:suff-jobs-nocommit}
 	Let $\alpha = 1$ and $\beta = \frac{\eps}{4}$. Then the \region completes at least half of all admitted jobs before their deadline.
\end{theorem}
\noindent The intuition for setting $\alpha=1$ and thus reserving regions of minimum size $|R(j)| = p_j$, for any $j$, is that in the model without commitment, a job does not need to block extra time in the future, (a  region larger than the job) 
to ensure the completion of preempted jobs.


%
The scheduling routine SPT guarantees the following. 
\begin{observation}\label{obs:b<=d}
	For $\alpha \geq 1$, the \region always prioritizes a job within its own region.
\end{observation}
Hence, for~$\alpha=1$, every job~$j$ completes at the end of the region at~$b_j$. Thus,~$j$ completes on time if and only if the region~$R(j)$ ends before~$d_j$, i.e.,~$b_j \leq d_j$. We prove Theorem~\ref{thm:suff-jobs-nocommit} by showing that at least half of all regions end before the deadline of their respective jobs. 


\begin{restatable}{lemma}{lemSmallestInstance} 
\label{lem:SmallestInstance}
	W.l.o.g.\ we restrict to instances on which the \region admits all jobs.
\end{restatable}

We want to show that the existence of a late job~$j$ implies that the subtree $T_j$ rooted in $j$ contains more finished than unfinished jobs. 
We fix a job~$j \in J$ that was admitted by the \region at~$a_j$ and whose region completes at~$b_j$. We want to analyze the structure of all regions~$\mathcal{R}$ in~$[a_j,b_j)$, i.e., the regions of all jobs in $T_j$. 
Let~$F_j$ denote the set of jobs in~$T_j$ that {\em finish on time}. Similarly, we denote the set of jobs in~$T_j$ that complete after their deadlines, i.e., that are {\em unfinished at their deadline}, by~$U_j$. 

\begin{restatable}{lemma}{lemFU}\label{lem:F>U}
	Let $\alpha =1$ and $\beta = \frac{\eps}{4}$, with $\eps >0$. If~$b_j - a_j \geq (\ell + 1) p_j$ for~$\ell > 0$, then~$|F_j| - |U_j| \geq \lfloor \frac{4\ell}{\eps}\rfloor$. 
\end{restatable}

\begin{proof}[Proof sketch.]
	Assume for the sake of contradiction that there is an instance such that the interruption tree generated by the \region contains a subtree~$T_j$ with~$b_j - a_j \geq (\ell + 1) p_j$ and~$|F_j| - |U_j| < \lfloor \frac{4\ell}{\eps}\rfloor$. Let~$\I$ be such an instance that uses a minimal number of jobs in total. 
	The goal is to construct an instance~$\I'$ that satisfies $b_j - a_j \geq (\ell + 1) p_j$ and~$|F_j| - |U_j| < \lfloor \frac{4\ell}{\eps}\rfloor$ although it uses less jobs than~$\I$. 
	
	To this end, we modify~$\I$ in several steps such that we can merge three jobs to one larger job without violating~$b_j - a_j \geq (\ell +1) p_j$, changing~$|F_j|$ or~$|U_j|$, or making the instance infeasible. The three jobs will be leaves with the same parent~$\ihat$ in~$T_j$. 
	If~$i$ is an unfinished job that has children which are all leaves, then 
	there have to be at least three jobs that interrupt~$\ihat$.		
	After merging the three jobs, we adapt the release date and deadline of~$\ihat$ to guarantee that the modified instance remains feasible. For all these modification steps, it is crucial that we can restrict to instances in which all jobs appear in the interruption tree (Lemma~\ref{lem:SmallestInstance}).
	
	However, this modification might lead to~$b_{\ihat} \leq d_{\ihat}'$ which implies that~$\ihat$ finishes on time. This changes the values of~$|F_j|$ and~$|U_j|$. Clearly, in this case,~$|U_j'| = |U_j| -1$. By a careful analysis, we see that the number of finished jobs decreases by one as well because the three children of~$\ihat$ are replaced by only one finished job. Hence, $|F_j'| - |U_j'| = |F_j| - |U_j|$. If~$\ihat$ does not finish by~$d_{\ihat}'$, then $|F_j'| - |U_j'| = (|F_j| - 2) - |U_j|$. Thus, the modified instance~$\I'$ also violates $|F_j'| - |U_j'| \geq \lfloor \frac{4\ell}{\eps}\rfloor$ but uses less jobs than~$\I$ does; a contradiction. 
\end{proof}


\begin{proof}[Proof of Theorem~\ref{thm:suff-jobs-nocommit}]
	Let $U$ be the set of jobs that are unfinished by their deadline but whose ancestors~(except machine job $M$) have all completed on time. Every job $j\in U$ was admitted by the algorithm at some time $a_j$ 
	with $d_j - a_j \geq (1+\delta) p_j$. With $\delta = \frac{\eps}{2}$ this implies $b_j - a_j > d_j - a_j \geq (1+\frac{\eps}{2}) p_j$. By Lemma \ref{lem:F>U} follows that  $|F_j| - |U_j| \geq \lfloor \frac{4\cdot \eps/2}{\eps} \rfloor = 2 $. Then, $
		|T_j| = |F_j| + |U_j| \leq 2|F_j| - 2 < 2|F_j|.
	$
 	This completes~the~proof. 
\end{proof}
	
\subsection{Scheduling with commitment}\label{subsec:ComAtAdmission}

We analyze the \region for scheduling with commitment. For both models, commitment at admission and $\delta$-commitment, we show that there is a choice of $\alpha$ and $\beta$ such that every job that has started processing will be completed before its deadline.  Recall that we can restrict to analyzing the algorithm in the $\delta$-commitment model since it runs with $\delta=\frac{\eps}{2}$ for commitment~at~admission.

\begin{restatable}{theorem}{lemalljobscommit}\label{lem:all-jobs-commit}
Let $\eps,\delta>0$ be fixed with $\delta <\eps$. If $\alpha\geq 1$ and $0<\beta < 1$ satisfy the condition that 	\vspace*{-1ex}
\begin{equation}\label{cond:1} 
	\frac{\alpha-1}{\alpha} \cdot \bigg( 1+ \delta - \frac{\beta}{1-\beta} \bigg) \geq 1, \vspace*{-1ex}
\end{equation}
then any job $j$ that is admitted by the algorithm at time $a_j\leq d_j -(1+\delta)p_j$ will be finished by $d_j$.
\end{restatable}

\begin{proof}
	Consider a job $j$ that is admitted (and simultaneously accepted for completion) by time $a_j$. It holds that $d_j-a_j\ge (1+\delta) p_j$. We show that~$j$ receives at least~$p_j$ units of time within~$[a_j,d_j)$. 
	Let~$|R(k)|$ denote the total length of intervals in~$R(k)$, the region of job $k$. 
	
	Let~$D_{j}\subseteq T_{-j}$ be the set of jobs whose region delays the region of job~$j$, and has nonempty intersection with~$[a_j,d_j)$. Notice that a job~$k\in D_{j}$ can only be released after time~$a_j$. Let~$D_{j}'\subseteq D_{j}$ be the subset of jobs whose region is completely contained in~$[a_j,d_j)$ and~$D_{j}''=D_{j}\setminus D_{j}'$.  
	
	Consider~$D_{j}'$. Notice that~$\big| \bigcup_{k\in D_{j}'} R(k) \big| = \alpha \sum_{k\in D_{j}'} p_k.$
	Thus, within regions~$R(k)$ of jobs~$k\in D_{j}'$, an~$\frac{\alpha-1}{\alpha}$-fraction of the total time is available for processing job~$j$.
	
	Consider~$D_{j}''=\{j_1,j_2,\cdots,j_\ell\}$ and assume that~$p_{j_1}\ge p_{j_2}\ge\cdots\ge p_{j_\ell}$. Any interval $[a_{j_i},b_{j_i})$ 
	of such a job $j_i$ in~$D_{j}''$ contains~$d_j$. This implies that~$\pi(j_{i+1}) = j_i$ for $0 \leq i < \ell$ where~$j_0 := j$ for simplicity. Applying Lemma~\ref{lem:GeometricallyDecreasing} gives an upper bound on the total processing volume of jobs in~$ D_{j}''$, i.e., 
	%
	$
	\sum_{i=1}^{\ell} p_{j_i}
	\leq \frac{\beta}{1-\beta}\cdot p_j.
	$
	
	To determine the amount of time for processing~$j$ within~$[a_j,d_j)$, we first subtract the total processing time for jobs in~$D_{j}''$. 
	The remaining interval may be covered with regions of~$D_{j}'$ within which we can use an~$\frac{\alpha-1}{\alpha}$-fraction as shown above. Thus, the amount of time that we can process job~$j$ within~$[a_j,d_j)$ is at least
	\begin{equation*}
	\frac{\alpha-1}{\alpha} \cdot \bigg( \left(d_j-a_j\right) -
	\sum_{j_i\in D_{j}''} p_{j_i} \bigg) 
	\geq \frac{\alpha-1}{\alpha} \cdot \left( 1+ \delta - \frac{\beta}{1-\beta} \right) \cdot p_j,
	\end{equation*}
	where we also use the fact that 
	$d_j-a_j \ge (1+\delta) p_j$. This bound is now independent of the actual schedule. We can conclude, if~$\alpha$ and~$\beta$ satisfy Condition~\eqref{cond:1}, then job~$j$ can process for~$p_j$ units of time within~$[a_j,d_j)$ and completes before its deadline.	
\end{proof}

\section{Competitiveness: admission of sufficiently many jobs}
\label{sec:comp-ratio}

We show that the \region admits sufficiently many jobs, independently of the commitment model. 

\begin{theorem}\label{theo:CompetitivenessRegion}
	The number of jobs that an optimal (offline) algorithm can complete 
	on time is at most a factor $\lambda + 1$ larger than the number of jobs admitted by the \region,~where~$\lambda:= \frac{\eps}{\eps -\delta}\frac{\alpha}{\beta}$, for $0<\delta<\eps\leq 1$.
\end{theorem}


To prove the theorem, we fix an instance and an optimal algorithm \opt. We can assume that an optimal offline algorithm does not schedule any job it cannot complete before its deadline. Let $X$ be the set of jobs that \opt scheduled and the \region did not admit. Let $J$ denote the jobs that the region algorithm admitted. Then, $X\cup J$, the union of the jobs scheduled only by \opt and the jobs admitted by our algorithm, is a superset of the jobs scheduled by \opt. Thus, showing $|X|\leq \lambda |J|$ implies Theorem~\ref{theo:CompetitivenessRegion}.

To this end, we develop a charging procedure that assigns each job in $X$ to a unique job in $J$ such that each job $j\in J$ is assigned at most $\lambda = \frac{\eps}{\eps -\delta}\frac{\alpha}{\beta}$ jobs. For a job~\mbox{$j\in J$} admitted by the \region we define the subset $X_j \subset X$ based on release dates. Then, we inductively transform the laminar family $(X_j)_{j\in J}$ into a partition $(Y_j)_{j \in J}$ of $X$ with $|Y_j| \leq \lambda$ for all $j\in J$ in the proof of Lemma~\ref{lem:SoN}, starting with the leaves in the interruption tree as base case (Lemma~\ref{lem:SoL}). For the construction of $(Y_j)_{j \in J}$, we heavily rely on the key property (Volume Lemma~\ref{lem:VolLemma}) and Corollary~\ref{cor:IsolatedJobs}.

More precisely, for a job $j\in J$ let $X_j$ be the set of jobs $x \in X$ that were released in the interval $ [a_j,b_j)$ and 
satisfy $p_x < \beta p_{\pi(j)}$. 
Let $X_j^S := \{x\in X_j: p_x < \beta p_j\}$ and $X_j^B := X_j \setminus X_j^S$ denote the \textit{small} and the \textit{big} jobs, respectively, in~$X_j$. Recall that $[a_j,b_j)$ is the convex hull of the region~$R(j)$ of job~$j$ and includes the convex hulls of all descendants of~$j$ in the interruption tree, i.e., jobs in $T_j$. In particular,~$X_k \subset X_j$ if~$k\in T_j$.  

\begin{observation}\label{obs:XandXj}
	\mbox{ } 
	\begin{compactenum}
		\item\label{obs:XandXj:X} 
		Any job~$x\in X$ that is scheduled successfully by \opt and that is not admitted by the \region is released within the region of some job~$j\in J$, that is,~$\bigcup_{j \in J} X_j = X$. 
		\item\label{obs:XandXj:Xj} As the \region admits any job that is small w.r.t.~$j$ and released in $R(j)$, $X_j^S = \bigcup_{k: \pi(k) = j} X_k$.
	\end{compactenum}
\end{observation}

As observed above, to prove the main Theorem \ref{theo:CompetitivenessRegion}, it suffices to show that~$|X| \leq \lambda |J|$. By Observation \ref{obs:XandXj}, $X = X_M^S$ and, thus, it is sufficient to show that $|X_M^S| \leq \lambda |J|$. In fact, we show a stronger statement. We consider each job~$j\in J$ individually and prove that the number of small jobs in~$X_j$ 
is bounded, i.e., $|X_j^S| \leq \lambda \tau_j$, 
where~$\tau_j$ is the number of descendants~of~$j$, i.e., $\tau_j := |T_{-j}|$.

More precisely, the fine-grained definition of the sets~$X_j$ in terms of the release dates 
and the processing times 
allows us to show that any job~$j$ with~$|X_j| > \lambda (\tau_j +1) $ has {\em siblings}~$j_1,\ldots,j_k$ such that~$|X_j| + \sum_{i=1}^k|X_{j_i}| \leq \lambda ( \tau_j +1 + \sum_{i=1}^k (\tau_{j_i}+1))$. We call~$i$ and~$j$ siblings if they have the same parent in the interruption tree. Simultaneously applying this charging idea to {\em all} descendants of a job~$h$ already proves $|X_h^S| \leq \lambda \tau_h $ as $X_h^S = \bigcup_{j: \pi(j) = h} X_j$ by Observation \ref{obs:XandXj}.

We prove that this ``balancing" of~$X_j$ between jobs only happens between siblings $j_1,\ldots,j_k$ with the property that $b_{j_i} = a_{j_{i+1}}$ for $1\leq i < k $. We call such a set of jobs a {\em string} of jobs. The ellipses in Figure \ref{fig:Region:ScheduleAsTree} visualize the maximal strings of jobs. A job~$j$ is called {\em isolated} if it has the property that  $b_i \neq a_j$ and $b_j \neq a_i$ holds for all children $i\neq j$ of $\pi(j)$. 

The next lemma is a key ingredient for the proof of Theorem \ref{theo:CompetitivenessRegion}. When we talk about a subset of $J$, we index the jobs in the subset in order of increasing admission points~$a_j$. Conversely, 
for a subset of~$X$, we order the jobs in increasing order of completion times, $C_x^*$, in the optimal schedule.

\begin{lemma}[Volume Lemma]\label{lem:VolLemma} 
Let $f,\ldots,g \in J$ be jobs with a common parent in the interruption tree. Let $x\in \bigcup_{j=f}^g X_{j}$ such that \vspace*{-2ex} \begin{equation}\label{eq:VolumeCondition}\tag{V}
	\sum_{j=f}^g \sum_{\substack{y\in X_j: \\ C_y^* \leq C_x^*}} p_y \geq \frac{\eps}{\eps-\delta}(b_{g} - a_{f}) + p_x. \vspace*{-1ex}
	\end{equation}
	Then, $p_x \geq \beta p_{\nextj}$, where $\nextj\in J\cup\{M\}$ is the job whose region contains $b_{g}$, i.e., $b_{g}\in R(\nextj)$. 
\end{lemma}
\begin{proof}[Proof of the Volume Lemma]
	Let $f,\ldots,g$, $x$, and $\nextj$ as in the lemma. Since $x \in X$, the region algorithm did not accept $x$ at time $b_{g}$. There are two possible reasons for this behavior: either $p_x \geq \beta p_{\nextj}$ or $x$ was not available for admission at time $b_{g}$ anymore. 

	Assume for the sake of contradiction that $p_x < \beta p_{\nextj}$ and, thus, $d_x - b_{g} < (1 + \delta) p_x$. By assumption, $r_x \geq a_{f}$ and $d_x - r_x \geq (1 + \eps) p_x$. Hence, \[
		b_{g} - a_{f} \geq b_{g} - d_x + d_x - r_x 
		> -(1+\delta) p_x + (1 + \eps) p_x 
		 = (\eps - \delta)p_x.
	\]
	
	By \eqref{eq:VolumeCondition}, the volume \opt processes between $b_{g}$ and $C_x^*$ is at least $
		\frac{\delta}{\eps-\delta}  (b_{g} - a_{f}) + p_x $. By applying the above calculated lower bound, we get that \[ \frac{\delta}{\eps-\delta}  (b_{g} - a_{f}) + p_x \geq \delta p_x + p_x =  (1 + \delta) p_x
	\] and, hence, that 
	$C_x^* \geq b_{g} + (1 + \delta) p_x > d_x$, which contradicts that \opt is a feasible schedule. 
\end{proof}

%
%
%

The next corollary follows directly from the Volume Lemma applied to a string of jobs or to a single job~$j \in J$ (let $f=j=g$). To see this, recall that $X_j$ contains only jobs that are small w.r.t.~$\pi(j)$, i.e., all $x\in X_j$ satisfy $p_x < \beta p_{\pi(j)}$.  We will use the corollary repeatedly to generate~a~contradiction.
\begin{corollary}\label{cor:IsolatedJobs}
	Let $\{ f,\ldots, g \} \subset J$ be a string of jobs and  let $x\in \bigcup_{j=f}^g X_j$ satisfy the Volume Condition \eqref{eq:VolumeCondition}. Then, there exists a sibling~$\nextj\in J$ of $g$ in the interruption tree with $b_g = a_{\nextj}$.
\end{corollary}

A main part of our proof is to show 
\eqref{eq:VolumeCondition} only relying on 
$|X_j| > \lambda (\tau_j +1) $. The relation between processing volume and cardinality of job sets is possible due to the definition of $X_j$ based on $T_j$.
The following lemma serves as base case for strings of leaves as well as role model for non-isolated nodes in the interruption~tree.

\begin{restatable}{lemma}{lemStringOfLeaves}
	\label{lem:SoL} Let $\{f, \ldots, g\}\subset J$ be jobs at maximal distance from~$M$ such that $\sum_{j=f}^i |X_{j}| > \lambda (i+1 - f) $ holds for all $f\leq i\leq g$. If~$g$ is the last such job, there is a sibling $\nextj$ of $g$ with $b_{g} = a_{\nextj}$ and $\sum_{j = f }^{\nextj} |X_{j}| \leq \lambda (\nextj + 1 - f)$.
\end{restatable}

\begin{proof}[Proof sketch]
	Observe that $[a_{f}, b_{g}) = \bigcup_{j=f}^g R(j)$ because the leaves $f, \ldots, g$ form a string of jobs. Thus, by showing that there is a job $x \in  X_f^g := \bigcup_{j=f}^g X_{j}$ that satisfies \eqref{eq:VolumeCondition}, we prove the statement with the Volume Lemma. To this end, we show that for every job $f \leq j \leq g$ there exists a set~$Y_j$ such that the processing volume within $Y_j$ is sufficient to cover the interval $[a_{j},b_{j})$ at least $\epsfrac$ times. More precisely, $Y_f,\ldots,Y_g$~satisfy  
	\begin{center}
	\begin{enumerate*}[label=(\roman*)]
		\item\label{enum:proofSoL:YSubsetX} $\bigcup_{j=f}^g Y_j \subset X_f^g$, 
		\item\label{enum:proofSoL:Y=lambda}  $|Y_j| = \lambda$, and 
		\item\label{enum:proofSoL:YBigJobs}  $Y_j \subset\{x \in X_f^g: p_x \geq  \beta p_{j}\}$ for every $f\leq j \leq g$. 
	\end{enumerate*}
	\end{center}
	Then, \ref{enum:proofSoL:Y=lambda} and \ref{enum:proofSoL:YBigJobs} imply $\sum_{y \in Y_j} p_y \geq \lambda \beta p_{j} = \epsfrac (b_{j} - a_{j})$. 
	Thus, if we choose~$x$ among those jobs in~$X_f^g$ that \opt completes last and guarantee that $x \notin \bigcup_{j=f}^g Y_j$, the Volume Condition \eqref{eq:VolumeCondition} is satisfied. We first describe
	 how to find $Y_f,\ldots,Y_g$ before we show that these sets satisfy \ref{enum:proofSoL:YSubsetX} to \ref{enum:proofSoL:YBigJobs}. 
	
	By assumption, $|X_{f}| > \lambda$. Index the jobs in $X_{f} = \{x_1,\ldots,x_\lambda,x_{\lambda+1},\ldots\}$ in increasing completion times $C_x^*$. Define $Y_f :=  \{x_1,\ldots,x_\lambda\}$ and $L_f := X_f \setminus Y_f $, i.e., $Y_f$ contains the $\lambda$ jobs in~$X_{f}$ that \opt completes first and~$L_f$ contains the last jobs. Let $Y_f,\ldots,Y_j$ and $L_j$ be defined for $f < j+1\leq g$. By assumption, $|X_{j+1 } \cup L_j| > \lambda$ since $|Y_i| = \lambda$ for $f\leq i \leq j$. We again index the jobs in $X_{j +1} \cup L_j = \{x_1, \ldots, x_\lambda, x_{\lambda+1}, \ldots\}$ in increasing optimal completion times. Then, $Y_{j+1} := \{x_1,\ldots,x_\lambda\}$ and $L_{j+1} :=  \{x_{\lambda+1},\ldots\}$. Since we move jobs only horizontally to later siblings, we call this procedure \pushforward.
	
	By definition, \ref{enum:proofSoL:YSubsetX} and \ref{enum:proofSoL:Y=lambda} are satisfied. Since $f,\ldots,g$ are leaves, the jobs in $Y_j \cap X_j$ are big w.r.t.~$j$. Thus, it remains to show that the jobs in~$L_j$ are big w.r.t.\ the next job~$j+1$. To this end, we assume that the jobs in $Y_f,\ldots,Y_j$ are big w.r.t.\ $f,\ldots,j$, respectively. If we find an index~$f\leq \ihat(x) \leq j$ such that~$x$ as well as the jobs in $\bigcup_{i=\ihat(x)}^j Y_i$ are released after~$a_{\ihat(x)}$
	and~$x$ is completed after every $y \in \bigcup_{i = \ihat(x)}^j Y_i$, 
	then the Volume Lemma \ref{lem:VolLemma} implies that $x \in L_j$ is big w.r.t.~$j+1$.  Indeed, then 	
	\(
		\sum_{i = \ihat(x)}^j \sum_{\substack{y\in X_{i}:  C_y^* \leq C_x^*}} p_y 		
		\geq p_x + \sum_{i = \ihat(x)}^j \sum_{y\in Y_i} p_y 	\geq  \epsfrac (b_{j} - a_{\ihat(x)}) + p_x .
	\)	
	We show by induction that such an index $\ihat(x)$ exists for every $x \in L_j$.

	As the same argumentation holds for $j = g$, Corollary \ref{cor:IsolatedJobs} implies the lemma. 
\end{proof}

\begin{restatable}{lemma}{lemSoN}\label{lem:SoN}
	For all $j\in J\cup\{M\}$, $|X_j^S| \leq \tau_j \lambda$. 
\end{restatable}

\begin{proof}[Proof sketch]
	We show that for every $j\in J\cup \{M\}$, there exists a partition $(Y_k)_{k\in T_{-j}}$ with 
	\begin{center}
		\begin{enumerate*}[label=(\roman*)]
		\item\label{enum:proofSoN:union} $\bigcup_{k\in T_{-j}} Y_k = X_j^S$,
		\item\label{enum:proofSoN:big} $Y_k \subset \{ x \in X_j: p_x \geq \beta p_k \}$, and
		\item\label{enum:proofSoN:size} $|Y_k| \leq \lambda$ for every $k\in T_{-j}$. 
		\end{enumerate*} 		
	\end{center}
	Then, it holds that $|X_j^S | = |\bigcup_{k\in T_{-j}} Y_k | = \sum_{k \in T_{-j}} |Y_k| \leq \lambda \tau_j $ and, thus, the lemma follows.
	
	The proof consists of an outer and an inner induction. The outer induction is on the distance~$\dist(j)$ of a job~$j$ from machine job~$M$, i.e., $\dist(M) := 0$ and $\dist(j) := \dist(\pi(j)) + 1$ for $j\in J$. The inner induction uses the idea about pushing jobs $x\in X_j$ to some later sibling of $j$ in the same string of jobs (see proof of Lemma~\ref{lem:SoL}).
	
	Let $j \in J$ with $\dist(j) = \dist_{\max} -1 := \max\{\dist(i): i \in J\} -1$. By Observation \ref{obs:XandXj}, $X_j^S = \bigcup_{k: \pi(k) = j} X_k$, where all $k \in T_{-j}$ are leaves at maximal distance from~$M$. We distinguish three cases for $k\in T_{-j}$:
	
	\begin{description}[labelindent=0em ,labelwidth=0cm, labelsep*=1em, leftmargin =\parindent , itemindent = 0pt, style = sameline,nosep]
		\item[Case 1.] If $k\in T_{-j}$ is isolated, $|X_k| \leq \lambda$ follows directly from the Volume Lemma as otherwise $\sum_{x \in X_k} p_x \geq \lambda \beta p_k + p_x = \tfrac{\eps}{\eps - \delta} (b_k - a_k) + p_x$ contradicts Corollary~\ref{cor:IsolatedJobs}, where~$x\in X_k$ is the last job that \opt completes from the set~$X_k$. Since all jobs in $X_k$ are big w.r.t.~$k$, we set $Y_k := X_k$.  
		\item[Case 2.] If $k\in T_{-j}$ with $|X_k| > \lambda$ is part of a string, let $f ,\ldots, g$ be the {\em maximal} string satisfying Lemma \ref{lem:SoL} with $k\in \{f,\ldots,g\}.$ With this lemma, we find $Y_f,\ldots,Y_g$ and set $Y_{g+1} := X_{g+1} \cup L_g$. 
		\item[Case 3.] We have not yet considered jobs~$k$ in a string with $|X_k| \leq \lambda$ that 
		do not have siblings~$f,\ldots,g$ in the same string 
		with~$b_g = a_k$ and $\sum_{i=f}^{g} |X_j| > (g - f) \lambda$. This means that such jobs do not receive jobs~$x \in X_i$ for~$i\neq k$ by the \pushforward procedure in Case~2. For such~$k\in T_{-j}$ we define~$Y_k := X_k$. 
	\end{description}	
	Then, $X_j^S = \bigcup_{k: \pi(k) = j} X_k = \bigcup_{k \in T_{-j}} X_k =  \bigcup_{k \in T_{-j}} Y_k$ and, thus, \ref{enum:proofSoN:union} to \ref{enum:proofSoN:size} are satisfied. 

	Let $\dist < \dist_{\max}$ such that $(Y_k)_{k\in T_{-j}}$ satisfying \ref{enum:proofSoN:union} to \ref{enum:proofSoN:size} exists for all $j\in J$ with $\dist(j) \ge \dist$. Fix~$j \in J$ with $\dist(j) = \dist -1$. By induction and Observation \ref{obs:XandXj}, it holds that $X_j^S =  \bigcup_{k: \pi(k) = j} \left(X_k^B \cup \bigcup_{i \in T_{-k}} Y_i \right) $.
	Now, we use the partitions $(Y_i)_{i \in T_{-k}}$ for $k$ with $\pi(k) = j$ as starting point to find the partition $(Y_k)_{k \in T_{-j}}$. 
	Fix $k$ with $\pi(k)= j$ and distinguish again the same three cases as before.
	
	\begin{description}[labelindent=0em ,labelwidth=0cm, labelsep*=1em, leftmargin = \parindent , itemindent = 0pt, style = sameline,nosep]
		\item[Case 1.] If $k$ is isolated, we show that $|X_k| \leq \lambda (\tau_k+1) $ and develop a procedure to find $(Y_i)_{i \in T_{k}}$. 
		By induction, $|X_k^S| \leq \lambda \tau_k $. In \ref{subsec:ProofofLem:SoN} we prove that $|X_k^B| \leq 	\lambda + (\lambda\tau_k  - |X_k^S|)$. To construct $(Y_i)_{i \in T_{k} }$, we assign $\min \{\lambda, |X_k^B|\}$ jobs from $X_k^B$ to $Y_k$. If $|X_k^B| > \lambda$, distribute the remaining jobs according to $\lambda - |Y_i|$ among the descendants of $k$. Then, $X_k = \bigcup_{i \in T_{k}} Y_i$. Because a job that is big w.r.t job $k$ is also big w.r.t.\ all descendants of $k$, 
		every (new) set $Y_i$ satisfies \ref{enum:proofSoN:big} and \ref{enum:proofSoN:size}. We refer to this procedure as \pushdown since jobs are shifted vertically to descendants.
		
		\item[Case 2.] If $|X_k| > \lambda (\tau_k+1) $, $k$ must belong to a string with similar properties as described in Lemma~\ref{lem:SoL}. This means, there is a maximal string of jobs $f,\ldots,g$ containing $k$ such that $\sum_{j=f}^i |X_{j}| > \lambda\sum_{j=f}^i \tau_j $ holds for all $f\leq i\leq g$ and $b_{j} = a_{j+1}$ for all $f\leq j < g$. 
		
		If the Volume Condition~\eqref{eq:VolumeCondition} is satisfied, there exists another sibling $g+1$ that balances the sets $X_f,\ldots,X_g,X_{g+1}$ due to Corollary \ref{cor:IsolatedJobs}. 
		This is shown by using \pushdown within a generalization of the \pushforward procedure. As the jobs $f,\ldots,g$ may have descendants, we use \pushforward to construct the sets $Z_f,\ldots,Z_g$ and $L_f,\ldots, L_g$ with $|Z_k| = \lambda (\tau_k+1)$. Then, we apply \pushdown to $Z_k$ and $(Y_i)_{i\in T_{-k}}$ in order to obtain $(Y_i)_{i \in T_{k} }$ such that they will satisfy $Z_k = \bigcup_{i\in T_{k}} Y_i$,
		$Y_i \subset \{ x \in X_j: p_x \geq \beta p_i \}$ and $|Y_i| = \lambda$ for every $i\in T_{k}$. Thus, the sets~$X_k$ with~$f\leq k\leq g$ satisfy~\eqref{eq:VolumeCondition} and we can apply~Corollary~\ref{cor:IsolatedJobs}.  

		\item[Case 3.] Any job $k$ with $\pi(k) = j$ that is part of a string and was not yet considered must satisfy $|X_{k} | \leq \lambda (\tau_k+1) $. 
			We use the \pushdown procedure for isolated jobs to get the partition~$(Y_{i})_{i \in T_k }$.
\end{description}
	Hence, we have found $(Y_k)_{k\in T_{-j}}$ with the properties \ref{enum:proofSoN:union} to \ref{enum:proofSoN:size}.	
\end{proof}

\begin{proof}[Proof of Theorem~\ref{theo:CompetitivenessRegion}]
	As explained before, the job set scheduled by \opt clearly is a subset of $X\cup J$, the union of jobs only scheduled by \opt and the jobs admitted by the \region. Thus, it suffices to prove that~$|X| \leq \lambda |J|$. By Observation \ref{obs:XandXj}, $X = X_M^S$ and, hence, $|X_M^S| \leq \lambda |J|$ implies~$|X| \leq \lambda |J|$. This is true as Lemma \ref{lem:SoN} also holds for the machine job~$M$.
\end{proof}


\paragraph{Finalizing the proofs of Theorems \ref{thm:UB-nocommitment} and \ref{thm:UB-commitment} }

\begin{proof}[Proof of Theorem \ref{thm:UB-nocommitment}]
	Set~$\alpha =1$ and $\beta = \frac{\eps}{4}$. Theorem~\ref{thm:suff-jobs-nocommit} shows that our algorithm completes at least half of all admitted jobs on time. 
	Theorem~\ref{theo:CompetitivenessRegion} implies that the \region is $\frac{16}{\eps}$-competitive.
\end{proof}

\begin{proof}[Proof of Theorem \ref{thm:UB-commitment}]
	By Lemma \ref{lem:all-jobs-commit}, the choice~$\alpha = \frac{8}{\delta}$ and~$\beta = \frac{\delta}{4}$ implies that the \region completes all admitted jobs. Theorem~\ref{theo:CompetitivenessRegion} implies that our algorithm is  ($\frac{32}{(\eps -\delta)\delta^2}+1$)-competitive.
\end{proof}

\section{Lower bounds on the competitive ratio}\label{sec:lower-bounds}

In this section we give a collection of lower bounds on the competitive ratio in the different commitment models and for different problem settings. To simplify notation, we formally introduce the notion of {\em laxity}. Let~$j$ be a job with processing time~$p_j$, deadline~$d_j$, and~$r_j$. The laxity~$\ell_j$ is defined as $d_j - r_j - p_j$. 


\paragraph{Scheduling without commitment.} 	We give a lower bound matching our upper bound 
in Theorem~\ref{thm:UB-commitment}. 

\begin{theorem}\label{thm:detLB-w=1-anytime}
	Every deterministic online algorithm has a competitive ratio $\Omega(\frac{1}{\eps})$. 
\end{theorem}

The proof idea is as follows: We release $\Omega(\frac{1}{\eps})$ 
	\emph{levels} of jobs. In each level, the release date of any but the first job is the deadline of the previous job. Whenever an online algorithm decides to complete a job from level~$i$ (provided no further jobs are released), then the release of jobs in level~$i$ stops and a sequence of $\OO(\frac{1}{\eps})$ jobs in level~$i+1$ is released. Jobs in level~$i$ have processing time~$p_i=2\eps\cdot p_{i-1}$ which is too large to fit in the slack of the already started job. Thus, an algorithm has to discard the job started at level $i$ to run a job of level $i+1$, meaning that it can only finish one job, while the optimum can finish a job from every other~level.
	
	We now give the formal proof.

\begin{proof}
	Let $\eps<\frac1{10}$ such that $\frac{1}{8\eps} \in \N$ and suppose there is an online algorithm with competitive ratio $c<\frac{1}{8\eps}$, from which it is sufficient to deduce a contradiction.
	We construct an adversarial instance in which each job $j$ belongs to one of $2\cdot\lceil c+1\rceil$ \emph{levels} and fulfills $d_j=r_j+(1+\eps)\cdot p_j$.
	The processing times are identical across any level, but they are decreasing by a factor of $2\eps$ when going from any level to the next.
	This (along with the interval structure) makes sure that no two jobs from consecutive levels can both be completed by a single schedule, which we will use to show that the online algorithm can only complete a single job throughout the entire instance.
	The decrease in processing times between levels, however, also makes sure that the optimum can finish a job from every other level, resulting in an objective value of $\lceil c+1\rceil$, which is a contradiction to the algorithm being $c$-competitive.
			
	The sequence starts in level $0$ at time $0$ with the release of one job~$j$ with processing time $p^{(0)}=1$ and, thus, deadline $d_{j}=1+\eps$. 
	We will show inductively that, for each level $i$, there is a time $t_i$ when there is only a single job $j_i$ left that the algorithm can still finish, and this job is from the current level $i$ (and, thus, $p_{j_i} = p^{(i)} = (2\eps)^{i}$).
	We will also make sure that at $t_i$ at most a $(\frac23)$-fraction of the time window of $j_i$ has passed.
	From $t_i$ on, no further jobs from level $i$ are released, and jobs from level $i+1$ start being released (or, if $i=2\cdot\lceil c+1\rceil-1$, we stop releasing jobs altogether). 
	It is clear that $t_0$ exists.
	
	Consider some time $t_i$, and we will release jobs from level $i+1$ so as to create time $t_{i+1}$. 
	The first job $j$ from level $i+1$ has release date $t_i$ and, by the above constraints, $d_j=t_i+(1+\eps)\cdot p_j$ where $p_j= p^{(i+1)}=(2\eps)^{i+1}$.
	As long as no situation occurs that fits the above description of $t_{i+1}$, we release an additional job of level $i+1$ at the deadline of the previous job from this level (with identical time-window length and processing time).
	We show that we can find time $t_{i+1}$ before $\frac{1}{8\eps}$ jobs from level $i+1$ have been released.
	Note that the deadline of the $\frac{1}{8\eps}$-th job from level $i+1$ is $t_i+\frac{1}{8\eps} \cdot (1+\eps) \cdot 2\eps\cdot p^{(i)}$, which is smaller than the deadline of $d_{j_i}$ since by induction $d_{j_i}-t_i\geq \frac23 \cdot p^{(i)}$ and $\eps<\frac1{10}$. This shows that, unless more than $\frac{1}{8\eps}$ jobs from level $i+1$ are released (which will not happen as we will show), all time windows of jobs from layer $i+1$ are contained in that of $j_i$.
	
	Note that there must be a job $j^\star$ among the $\frac{1}{8\eps}$ first ones in level $i+1$ that the algorithm completes if no further jobs are released within the time window of $j^\star$: By induction, the algorithm can only hope to finish a single job released before time $t_i$ and the optimum could complete $\frac{1}{8\eps}$ jobs from level $i+1$, so $j^\star$ must exist for the algorithm to be $c$-competitive. Now we can define $j_{i+1}$ to be the first such job $j^\star$ and find $t_{i+1}$ within its time window: At the release date of $j^\star$, the algorithm could only complete $j_i$. However, since the algorithm finishes $j_{i+1}$ if there are no further jobs released, and $\eps<\frac1{10}$, it must have worked on $j_{i+1}$ for more than $\frac{p^{(i+1)}}2$ units of time until $r_{i+1}+\frac23\cdot p^{(i+1)}=:t_{i+1}$. This quantity, however, exceeds the laxity of $j_i$, meaning that the algorithm cannot finish $j_i$ any more. (Recall that the laxity of~$j_i$ is $\eps p^{(i)} = 2^i \eps^{i+1}$.) So $t_{i+1}$ has the desired properties.
			
	This defines $t_{2\cdot\lceil c+1\rceil}$, and indeed the algorithm will only finish a single job. We verify that an optimal algorithm can schedule a job from every other level. Note that, among levels of either parity, processing times are decreasing by a factor of $4\eps^2$ between consecutive levels. So, for any job $j$, the total processing time of jobs other than $j$ that need to be processed within the time window of $j$ adds up to less than 
	$$\sum_{i=1}^\infty (4\eps^2)^i\cdot p_j = 4\eps^2 \cdot \sum_{i=0}^\infty (4\eps^2)^i \cdot p_j = 4\eps^2 \cdot \frac1{1-4\eps^2} \cdot p_j \leq \eps \cdot \frac{4}{10} \cdot \frac1{1-\frac4{100} } \cdot p_j <  \eps\cdot p_j=\ell_j.$$
	 This completes the proof.
\end{proof}

\noindent {\bf Commitment upon arrival.} We substantially strengthen earlier results for weighted jobs~\cite{DBLP:conf/spaa/LucierMNY13,Yaniv-Thesis2017} and show that the model is hopeless even in the unweighted setting and even for randomized algorithms.

\begin{theorem}\label{thm:LBCommUponReleasewj=1}
	No randomized online algorithm has a bounded competitive ratio 
	for commitment upon arrival.
\end{theorem}

In the proof of the theorem, we use the following algebraic fact.

\begin{lemma}\label{lem:algebraic-fact}
	Consider positive numbers $n_1,\dots,n_k, c\in\mathbb{R}_+$ with  the following properties:
	\begin{compactenum}[(i)]
		\item\label{item:algebraic-fact-ub} $\sum_{i=1}^k n_i\leq 1$,
		\item\label{item:algebraic-fact-lb} $\sum_{i=1}^j n_i\cdot 2^{i-1}\geq \frac{2^{j-1}}{c}$ for all $j=1,\dots,k$.
	\end{compactenum}
	Then $c\geq\frac{k+1}{2}$.
\end{lemma}
\begin{proof}
	We take a weighted sum  over all inequalities in~(\ref{item:algebraic-fact-lb}), where the weight of the inequality corresponding to $j<k$ is $2^{k-j-1}$ and the weight of the inequality corresponding to $j=k$ is $1$. The result is $$\sum_{i=1}^k n_i\cdot 2^{k-1}\geq\frac{(k+1)\cdot 2^{k-2}}{c}\;\Leftrightarrow\;\sum_{i=1}^k n_i\geq\frac{(k+1)}{2c}.$$ If $c<\frac{k+1}{2}$, this contradicts~(\ref{item:algebraic-fact-ub}).
\end{proof}

We proceed to the proof of the theorem.

\begin{proof}[Proof of Theorem~\ref{thm:LBCommUponReleasewj=1}]
	Consider any $\eps>0$ and arbitrary $\gamma\in(0,1)$.
	Suppose there is a (possibly randomized) $c$-competitive algorithm, where $c$ may depend on $\eps$.
	
	We will choose some $k\in\mathbb{N}$ later. The adversary releases at most $k$ \emph{waves} of jobs, but the instance may end after any wave.
	Wave $i$ has $2^i$ jobs. Each job from the $i$-th wave has release date $\frac i k\cdot\gamma$, deadline $1$, and processing time $\frac{1}{2^i}\cdot\frac{1-\gamma}{1+\eps}$.
	Note that choosing $p_j\leq\frac{1-\gamma}{1+\eps}$ for all jobs $j$ makes sure that indeed $\ell_j\geq\eps\cdot p_j$, and observe that the total volume of jobs in wave $i$ adds up to no more than $1-\gamma$.
	
Define $n_i$ to be the expected total processing time of jobs that the algorithm accepts from wave $i$. We observe:
\begin{compactenum}[(i)]
	\item Since all accepted jobs have to be scheduled within the interval $[0,1]$, we must have $\sum_{i=1}^k n_i\leq 1$.
	\item For each $i$, possibly no further jobs are released after wave $i$. Since, in this case, the optimum schedules all jobs from wave $i$ and the jobs' processing times decrease by a factor of $2$ from wave to wave, it must hold that $\sum_{i=1}^j n_i\cdot 2^{i-1}\geq \frac{2^{j-1}}{c}$.
\end{compactenum}
 This establishes the conditions necessary to apply Lemma~\ref{lem:algebraic-fact} to $n_1,\dots,n_k$, which shows that choosing $k\geq 2c$ yields a contradiction.
\end{proof}


\noindent {\bf Commitment on job admission and $\delta$-commitment.} Since these models are more restrictive than scheduling without commitment, the lower bound $\Omega(\frac{1}{\eps})$ from Theorem~\ref{thm:detLB-w=1-anytime} holds. 
	In the present setting we can provide a much simpler (but asymptotically equally strong) lower bound.

{\em Commitment upon admission.} For scheduling with arbitrary weights, Azar et al.~\cite{AzarKLMNY15} 
rule out any bounded competitive ratio for deterministic algorithms. 
Thus, our bounded competitive ratio for the unweighted setting (Theorem~\ref{thm:UB-commitment}) gives a clear separation between the  weighted and the unweighted setting. 

{\em Scheduling with $\delta$-commitment.} We give a lower bound depending on parameters $\eps$ and $\delta$.

\begin{theorem}\label{thm:LB-dcommit}
	Consider scheduling weighted jobs in the $\delta$-commitment model. For any $\delta>0$ and any $\eps$ with $\delta \leq \eps < 1+\delta$, no deterministic online algorithm has a bounded competitive ratio.
\end{theorem}
\begin{proof}
	 We reuse the idea of \cite{AzarKLMNY15} to release the next job upon admission of the previous one while heavily increasing the weights of subsequent jobs. However, the scheduling models differ in the fact that the~$\delta$-commitment model allows for processing before commitment which is not allowed in the commitment-upon-admission model.  
	 
	Assume for the sake of contradiction, that there is a~$c$-competitive algorithm. We consider the following instance that consists of~$n$ tight jobs with the same deadline~$d := 1+\eps$. Job~$j$ has a weight of~$w_j := (c+1)^j$ which implies that any~$c$-competitive algorithm has to admit job~$j$ at some point even if all jobs~$1,\ldots,j-1$ are admitted. In the~$\delta$-commitment model, the admission cannot happen later than~$d-(1+\delta) p_j$ which is the point in time when job~$j+1$ is released. 
	
	More precisely, the first job is released at~$r_1=0$ with processing time~$p_1=1$. If jobs~$1,\ldots,j$ have been released, job~$j+1$ is released at~$r_{j+1} = d-(1+\delta) p_j$ and has processing time 
	\begin{equation*}
		\frac{d - r_{j+1} }{1+\eps} = \frac{d - (
		d-(1+\delta) p_j)}{1+\eps} = \frac{1+\delta}{1+\eps} p_j = \ldots = \left( \frac{1+\delta}{1+\eps} \right)^{j-1}.
	\end{equation*}
	An instance with~$n$ such jobs has a total processing volume of 
	\begin{equation*}
		\sum_{j=1}^n p_j = \sum_{j=0}^{n-1}\left( \frac{1+\delta}{1+\eps} \right)^{j} = \frac{1 -  \left(\frac{ 1+\delta }{1+\eps}\right)^{n} } { 1 - \frac{ 1+\delta }{1+\eps} }.
	\end{equation*}
	Any~$c$-competitive algorithm has to complete the~$n$ jobs before~$d = 1+\eps$. This also holds for~$n \rightarrow \infty$ and, thus, $\frac{1 + \eps } { \eps - \delta} \leq 1+ \eps$ is implied. This is equivalent to~$\eps \geq 1+ \delta$. In other words, if~$\eps < 1 + \delta$, there is no deterministic~$c$-competitive online algorithm. 
\end{proof}

 In particular, there is no bounded competitive ratio possible for $\eps\in(0,1)$. A restriction for $\eps$ appears to be necessary 
	as Azar et al.~\cite{AzarKLMNY15} provide 
	such a bound when the slackness is sufficiently large, i.e, $\eps>3$. In fact, our bound answers affirmatively the open question in~\cite{AzarKLMNY15} if high slackness is indeed required. Again, this strong impossibility result gives a clear separation between the weighted and the unweighted problem as we show in the unweighted setting a bounded competitive ratio for any $\eps>0$ (Theorem~\ref{thm:UB-commitment}).
	
\paragraph{Further lower bounds.} In Appendix~\ref{app:lbs}, we provide a simple proof of Theorem~\ref{thm:detLB-w=1-anytime} for commitment on job admission and $\delta$-commitment and additional lower bounds for the setting of proportional weights ($p_j=w_j$).

\section{Concluding remarks}

We provide a general framework for online deadline-sensitive scheduling with and without commitment. This is the first unifying approach and we believe that it captures well (using parameters) the key design principles needed when scheduling {\em online, deadline-sensitive} and {\em with commitment}. 

We give the first rigorous bounds on the competitive ratio for maximizing throughput in different commitment models. Some gaps between upper and lower bounds remain and, clearly, it would be interesting to close them. In fact, the lower bound comes from scheduling without commitment and it is unclear, if scheduling with commitment is truly harder than without.
It is somewhat surprising that essentially the same algorithm (with the same parameters and commitment upon admission) performs well for both, commitment upon admission and $\delta$-commitment, whereas a close relation between the models does not seem immediate. It remains open, if 
an algorithm can exploit the seemingly greater flexibility~of~$\delta$-commitment.

We restrict our investigations to unit-weight jobs which is justified by strong impossibility results 
(Theorems \ref{thm:LBCommUponReleasewj=1}, \ref{thm:LB-dcommit}, \cite{Yaniv-Thesis2017,DBLP:conf/spaa/LucierMNY13,AzarKLMNY15}). Thus, for weighted throughput a rethinking of the model is needed. A major difficulty is the interleaving structure of time intervals which makes the problem intractable in combination with weights. However, practically relevant restrictions to special structures such as laminar or agreeable intervals 
have been proven to be substantially better tractable in related online deadline scheduling research~\cite{ChenMS16,ChenMS16b}.

Furthermore, we close the problem of scheduling unweighted jobs without commitment with a best-achievable competitive ratio $\Theta(\frac{1}{\eps})$. It remains open if the weighted setting is indeed harder than the unweighted setting or if the upper bound $\OO(\frac{1}{\eps^2})$ in~\cite{DBLP:conf/spaa/LucierMNY13} can be improved.
%
%
%
Future research on generalizations to multi-processors seems highly relevant. 
We believe that our general framework is a promising~starting~point. 


\newpage
\appendix
\section*{APPENDICES}

\section{Summary State-of-the-Art}

For convenience we summarize the state of the art regarding competitive analysis for online throughput maximization with and without commitment.

\begin{table}[h]
	\centering
	\def\arraystretch{1.1}
	\begin{tabular}{ccccc}
		\toprule
		&no commitment&commit at admission& $\delta$-commitment & commit at arrival\\
		\midrule
		\multirow{2}{*}{$w_j\equiv 1$}& $\Theta(\frac 1\eps)$ & $\Omega(\frac 1\eps),\mathcal{O}(\frac 1{\eps^2})$ & $\Omega(\frac 1\eps),\mathcal{O}(\frac{\eps}{(\eps-\delta)\delta^2})$ & no $f(\eps)$ \\
		&  Theorems \ref{thm:UB-nocommitment} and \ref{thm:detLB-w=1-anytime} & Theorems \ref{thm:UB-commitment} and \ref{thm:detLB-w=1-anytime}  &   Theorems  \ref{thm:UB-commitment} and \ref{thm:detLB-w=1-anytime}   & Theorem \ref{thm:LBCommUponReleasewj=1}\\[1ex]  
		\midrule
		\multirow{2}{*}{$w_j=p_j$}& $\mathcal{O}(1)$ & $\Theta(\frac1\eps)$ & $\Theta(\frac{1}{\eps})$ & $\Theta(\frac1\eps)$ \\
		&  \cite{DBLP:journals/siamcomp/KorenS95} & \cite{GarayNYZ02,
		DBLP:conf/approx/DasGuptaP00} & Theorem~\ref{thm:wj=pj:admission:LB} ,\cite{GarayNYZ02,DBLP:conf/approx/DasGuptaP00}& \cite{GarayNYZ02,DBLP:conf/approx/DasGuptaP00} \\[1ex]  
		\midrule
		\multirow{2}{*}{general $w_j$}&$\Omega(\frac{1}{\eps}),\mathcal{O}(\frac{1}{\eps^2})$ & no $f(\eps)$ & no $f(\eps$) & no $f(\eps)$\\
		&  Theorem \ref{thm:detLB-w=1-anytime}, \cite{DBLP:conf/spaa/LucierMNY13} & \cite{AzarKLMNY15} & Theorem \ref{thm:LB-dcommit} & \cite{DBLP:conf/spaa/LucierMNY13} \\[1ex]  		                     
		\bottomrule
	\end{tabular}
\end{table}




\section{Proofs of Section~\ref{sec:Feasibility} (Completing sufficiently many jobs)
}

\subsection{Proofs of Section~\ref{subsec:NoCommitment}}

\lemSmallestInstance*

More formally, we show the following lemma.

\begin{lemma}
	For any instance~$\I$ and some job~$j$, for which the \region generates an interruption tree~$T_j$ with regions in $[a_j, b_j)$, there is an instance~$\I'$ with at most~$|T_j|+1$ jobs such that the regions in~$[a_j,b_j)$ and the tree~$T_j$ are identical. 
\end{lemma}

\begin{proof}
	Consider an instance~$\I$ of the non-committed scheduling problem, and let~$T_j$ be the interruption tree constructed by the \region with its root in~$j$. Let the interval~$[a_j, b_j)$ be the convex hull of the subintervals belonging to the region of~$j$. Our goal is to modify the instance~$\I$ such that we can remove jobs outside of~$T_j$ without changing the interruption tree of the algorithm. We do so by setting~$\I'$ to contain the set of jobs in~$T_j$ and changing only parameters of job~$j$ (and possibly adding one auxiliary job). Note that~$j$ is, by definition, the largest job in~$\I'$. 
		
	If~$d_j - a_j \geq (1+\eps) p_j$ then we set~$r_j':= a_j$. Otherwise, 
	we add an auxiliary job~$0$ to~$\I'$ that is tight and blocks the machine until~$a_j$. This means~$r_0 = d_j - (1+\eps) p_j$, $p_0 = (1+\eps) p_j - (d_j - a_j)$, and~$d_0 = r_0 + (1+\eps) p_0$. Moreover, we modify the release date of~$j$ to~$r_j':= r_0$. Since the auxiliary job is the smallest job in instance~$\I'$ at time~$r_0$, the \region admits this job and delays job~$j$. Let~$\mathcal{R}$ and~$\mathcal{R}'$ be the schedule of regions in $[a_j,b_j)$ generated by the region algorithm when applied to~$\I$ and~$\I'$, respectively. We show that~$\mathcal{R}$ and~$\mathcal{R}'$ are identical in~$[a_j, b_j)$. 
	
	Consider the time~$t = a_j$. Clearly,~$j$'s region starts in~$\mathcal{R}$ by assumption. If no auxiliary job was used, job~$j$ is the only available job in~$\I'$. Thus, the \region admits~$j$. In contrast, if~$0 \in \I'$, it finishes at~$a_j$ by definition. Since~$j$ is admitted in~$\mathcal{R}$, it must hold that~$d_j - a_j \geq (1+\delta)p_j$. Thus, its regions also begins in~$\mathcal{R}'$ at~$a_j$. 
	
	Let~$a_j < t < b_j$ be the first time when the two region schedules~$\mathcal{R}$ and~$\mathcal{R}'$ are different. Since both schedules are generated by the \region, any change in the structure of the regions is due to one of the two decision events of the \region. Recall that these events where the end of a job's region and the release of a new job. We distinguish two cases based on the job~$k$ that caused the difference in~$\mathcal{R}$ and~$\mathcal{R}'$:~$k\in T_j$ or~$k\notin T_j$. 	
	
	By definition, the region of any job outside of~$T_j$ has empty intersection with~$[a_j, b_j)$. Thus, the release of such a job can neither change~$\mathcal{R}$ nor~$\mathcal{R}'$. Of course, the region of such job cannot end within~$[a_j, b_j)$. Thus, a job~$k\in T_j$ is the reason for the difference in~$\mathcal{R}$ and~$\mathcal{R}'$. Let~$k$ be the job that owns time~$t$ in~$\mathcal{R}$. If the processor is idle in~$\mathcal{R}$ after~$t$ let~$k$ be the job that owns~$t$ in~$\mathcal{R}'$. As the two schedules are identical in~$[a_j,t)$, let~$i \in T_j$ be the unique job that owns the time right before~$t$. 
	
	Consider the event that the region of~$i$ ended at~$t$. If~$k$ is an ancestor of~$i$, then there is no sufficiently small job available in~$\I$ that prevents~$k$ from being restarted at~$t$. Additionally, the amount of time that belongs to~$R(k)$ in $[a_j,t)$ is identical in~$\mathcal{R}$ and~$\mathcal{R}'$. Thus,~$k$ also resumes processing in~$\mathcal{R}'$. If~$k$ is admitted at~$t$ in~$\mathcal{R}$, it is sufficiently small to (further) preempt the ancestor of~$i$ and it is available for admission. Hence, these two properties are also satisfied in~$\I'$ and~$k$ is admitted at~$t$ in~$\mathcal{R}'$ as well. Therefore~$\mathcal{R}$ contains idle time at~$t$ while the region of~$k$ is scheduled in~$\mathcal{R}'$. Since the jobs in~$\I'$ are a subset of the jobs in~$\I$ (except for $0$), job~$k$ is also admitted and unfinished or available at~$t$ in~$\I$. This is a contradiction.	
\end{proof}


\lemFU*

\begin{proof}
	
	Assume for the sake of contradiction that there is an instance such that the interruption tree generated by the \region contains a subtree~$T_j$ with~$b_j - a_j \geq (\ell + 1) p_j$ and~$|F_j| - |U_j| < \lfloor \frac{4\ell}{\eps}\rfloor$. Let~$\I$ be such an instance with minimal number of jobs in total. By Lemma \ref{lem:SmallestInstance}, we restrict to the jobs contained in~$T_j$. The goal is to construct an instance~$\I'$ that satisfies $b_j - a_j \geq (\ell + 1) p_j$ and~$|F_j| - |U_j| < \lfloor \frac{4\ell}{\eps}\rfloor$ although it uses less jobs than~$\I$. To this end, we modify~$\I$ in several steps such that we can merge three jobs to one larger job without violating~$b_j - a_j \geq (\ell +1) p_j$, changing~$|F_j|$ or~$|U_j|$, or making the instance infeasible. The three jobs will be leaves with the same parent~$\ihat$ in~$T_j$. In fact, if~$\ihat$ is an unfinished job, then $b_{\ihat} - a_{\ihat} \geq (1 +\frac{\eps}{2}) p_{\ihat}$. Any job~$k$ that may postpone~$\ihat$ satisfies~$p_k < \beta p_{\ihat} = \frac{\eps}{4} p_{\ihat}$. Hence, if the children of~$i$ are all leaves, there have to be at least three jobs that interrupt~$\ihat$. 

	
	In the following argumentation we use the position of jobs in the interruption tree. To that end, we define the height of an interruption tree to be the length of a longest path from root to leaf and the height of the node~$j$ in the tree to be the height~of~$T_j$.

	The roadmap of the proof is as follows: First, we show that there is an instance~$\I'$ with an unfinished job of height one by proving the following facts.
	\begin{enumerate}
		\item\label{enum:FU:HeightofTj} The height of the interruption tree~$T_j$ is at least two. 
		\item Any finished job of height one can be replaced by a leaf. 
	\end{enumerate}
	Let~$\ihat$ be an unfinished job of height one. We show that three of its children can be merged and become a sibling of their former ancestor. To this end, we prove that we can assume w.l.o.g. the following two properties of the children of~$\ihat$. 
	\begin{enumerate}[resume*]
		\item No job is completely scheduled in~$[d_{\ihat}, b_{\ihat})$. 
		\item The children of~$\ihat$ form a string of jobs.
	\end{enumerate}
	
	\begin{enumerate}[label=Ad \arabic*.]
		\item We show that the height of~$T_j$ is at least two. Assume for the sake of contradiction that~$T_j$ is a star centered at~$j$. Since any leaf finishes by definition of the \region, the root~$j$ is the only job that could possibly be unfinished. As~$|F_j| - |U_j|  \leq \lfloor \frac{4\ell}{\eps}\rfloor - 1$, this implies that there are at most~$\lfloor \frac{4\ell}{\eps}\rfloor$ leaves in~$T_j$. Then, \begin{equation*}
		b_j - a_j = \sum_{i \in T_j} p_i < p_j + \left\lfloor \frac{4\ell}{\eps}\right\rfloor \cdot \frac{\eps}{4}p_j \leq p_j + \ell p_j,
		\end{equation*}
		where we used~$p_i < \beta p_j$ for each leaf~$i \in T_j$. This contradicts the assumption~$b_j - a_j \geq (\ell +1) p_j$. 
		\item 	We show that w.l.o.g. any region of height one ends after the corresponding deadline. Let~$i$ be a finished job of height one, i.e.,~$b_i \leq d_i$, let~$l$ be the last completing child of~$i$, and let~$\pi(i)$ denote the parent of~$i$. This job~$\pi(i) \in T_j$ must exist because the height of~$T_j$ is at least two by \ref{enum:FU:HeightofTj}. We create a new instance by replacing~$l$ by a new job~$l'$ that is released at~$r_{l'} := b_i - p_l$ and that has the same processing time as~$l$, i.e.,~$p_{l'}:= p_l$. The deadline is~$d_{l'} := r_{l'}+(1+\eps)p_{l'}$. We argue that~$i$ finishes at~$r_{l'}$ in the new instance and that~$l'$ finishes at~$b_i$. 
		
		As~$l$ is not interrupted, $r_{l'} - a_l = b_i - p_l - a_l = b_i - b_l $, which is the remaining processing time of~$i$ at~$a_l$. If we can show that~$i$ is not preempted in ~$[a_l,r_{l'})$ in the new instance,~$i$ completes at~$a_l + b_i - b_l = r_{l'}$. Since~$l$ is the last child of~$i$, any job~$k$ released within~$[a_l,b_i)$ is scheduled later than~$b_i$. (Recall that we restrict to jobs in the interruption tree~$T_j$ by Lemma~\ref{lem:SmallestInstance}.) Thus,~$p_{k} \geq \beta p_i > p_l$. Hence,~$i$ is not interrupted in~$[a_l,r_{l'})$ and completes at~$r_{l'} < b_i \leq d_i$. At time~$r_{l'}$, job~$l'$ is the smallest available job and satisfies~$p_{l'} < (\frac{\eps}{4}) p_i < (\frac{\eps}{4})^2 p_{\pi(i)}$. Thus,~$l'$ is admitted at~$r_{l'}$ and is not interrupted until~$r_{l'} + p_l = b_i$ by the same argumentation about the jobs~$k$ that are released in~$[a_l,b_i)$. Hence, its region ends at~$r_{l'}+ p_{l'} < d_{l'}$. Moreover, outside the interval~$[a_l,b_i)$ neither the instance nor the schedule of the regions are changed. Since~$l'$ is now released outside of the region of~$i$,~$l'$ becomes a child of~$\pi(i)$, i.e.,~$l'$ is directly appended to~$\pi(i)$. This modification does not alter the length of~$[a_j,b_j)$ or the number of finished and unfinished jobs. Inductively applying this modification to any finished job of height one proves the claim. 
	\end{enumerate}

	Next, we prove the simplifying assumptions on jobs of height one. Because of the just proven statements,~$T_j$ must contain at least one unfinished job of height one. Let~$\ihat$ be such a job. Since~$\ihat$ is unfinished, it must hold that~$d_{\ihat} < b_{\ihat}$. For simplicity, let~$T := T_{-\ihat}$ denote the set of children of~$\ihat$ and let~$\tau:= |T|$.	
	\begin{enumerate}[resume*]
		\item We start by showing that no region of child of~$\ihat$ is completely contained in~$[d_{\ihat},b_{\ihat})$. If there is a child~$c$ with~$d_{\ihat} \leq a_c < b_c < b_{\ihat}$, it does not prevent the algorithm from finishing~${\ihat}$. Hence, it could be appended to~$\pi(\ihat)$ in the same way as we appended the last child of an finished job in the previous claim. That is, we can create a new instance in which~$c$ is appended to~$\pi(\ihat)$ and~$\ihat$ is still unfinished.
		\item We show that the regions of the children of~$\ihat$ form an interval with endpoint~$\max\{d_{\ihat}, b_{\max}\}$ where~$b_{\max} := \max_{c \in T} b_c$.  We further prove that they are released and admitted in increasing order of their processing times. More formally, we index the children in increasing order of their processing times, i.e,~$p_{c_1} \leq p_{c_2} \leq \ldots \leq p_{c_\tau}$. Then, we create a new instance with modified release dates such that one child is released upon completion of the previous child. This means~$r_{c_\tau'}:=\max\{d_i, b_{\max}\} - p_{c_\tau}$ and~$r_{c_{h-1}'} := r_{c_{h}'} - p_{c_{h-1}'}$ for~$1< h \leq \tau$ where the processing times are not changed, i.e.,~$p_{c_h'} = p_{c_h}$. In order to ensure that the modified instance is still feasible, we adapt the deadlines $d_{c_h'} := r_{c_h'} + (1+\eps) p_{c_h'}$.
		
		It is left to show that the modifications did not affect the number of finished or unfinished jobs. Obviously, the \region still admits every job in~$T'$. A job~$k \notin T'$ released in~$[a_{\ihat},b_{\ihat})$ satisfies~$p_{k} \geq \beta p_{\ihat} > p_c > \beta p_c$ for all~$c\in T$. Hence, these jobs do not interrupt either~$\ihat$ or any of the children. They are still scheduled after~$b_{\ihat}$ and every child~$c\in T$ completes before its deadline. We also need to prove that~$\ihat$ still cannot finish on time. 
		If~$b_{\max} \leq d_{\ihat}$, the region of every child is completely contained in~$[a_{\ihat}, d_{\ihat})$. Hence, job~$\ihat$ is still interrupted for the same amount of time before~$d_{\ihat}$ in~$\I'$ as it is in~$\I$. Thus, $b_{\ihat}'= a_{\ihat} + p_{\ihat} + \sum_{h=1}^{\tau} p_{c_h'} =  a_{\ihat} + p_{\ihat} + \sum_{c \in T} p_{c} = b_{\ihat} > d_{\ihat}$. 
		If~$b_{\max} > d_{\ihat}$, let $l$ denote the child in~$\I$ with~$b_l = b_{\max}$. Then,~$r_{c_\tau'} = b_{\max} - p_{c_\tau'} \leq b_l - p_l < d_{\ihat} $, where we used that no child is completely processed in~$[d_{\ihat}, b_{\ihat})$ and that~$c_\tau'$ is the child of~$\ihat$ with the largest processing time. Thus, the delay of~$\ihat$ in~$[a_{\ihat}, d_{\ihat})$ is identical to~$\sum_{c \in T} p_c - (b_l - d_{\ihat})$. Hence,~$\ihat$ still cannot finish on time. In this case,~$b_{\ihat}'= b_{\ihat}$ holds as well. Hence, the modified jobs in~$\I'$ still cover the same interval~$[a_{\ihat},b_{\ihat})$.
	\end{enumerate}
	
	So far we have modified $\I$ such that it remains an instance which achieves $|F_j|-|U_j| < \lfloor \frac{4\ell}{\eps}\rfloor$ with a minimum total number of jobs. In the following, we show that the considered instance does not use a minimal number of jobs in total which implies a contradiction and thus the lemma is proved. We do so by modifying the instance in three steps. In the first step, we merge three jobs in~$T_{-\ihat}$ where~$\ihat \in T_j$ is an unfinished job of height one such that its children satisfy the Assumptions~3 and~4. In the second step, we replace~$\ihat$ by a similar job~$\ihat'$ to ensure that the instance is still feasible. In the third step, we adapt jobs~$k \notin T_{-\ihat}$ to guarantee that~$\ihat'$ is admitted at the right point in time. Then, we we show that the resulting instance covers the same interval~$[a_{\ihat}, b_{\ihat})$

 	Since~$\ihat$ is admitted at~$a_{\ihat} \leq d_{\ihat} + (1+\frac{\eps}{2}) p_{\ihat}$ and not finished by the region algorithm on time,~$b_{\ihat}- a_{\ihat}\geq (1+\frac{\eps}{2}) p_{\ihat}$. As any job that may postpone the region of~$\ihat$ satisfies~$p_k < \beta p_{\ihat} = \eps /4 p_{\ihat}$, there have to be at least three jobs that interrupt~${\ihat}$. Among these, consider the first three jobs~$c_1,c_2,$ and~$c_3$ (when ordered in increasing release dates). We create a new instance by deleting $c_1,c_2,$ and~$c_3$ and adding a new job $c'$ such that~$c'$ is released at the admission date of~$\ihat$ in $\I$ and it merges~$c_1,c_2$ and~$c_3$, i.e.,~$r_{c'} := a_{\ihat}$,~$p_{c'} := p_{c_1} + p_{c_2} + p_{c_3}$, and~$d_{c'} := r_{c'}+(1+\eps)p_{c'}$. In exchange, we remove the jobs ~$c_1,c_2,$ and~$c_3$ from the instance. This completes our first step of modification. 
 	
 	Second, we replace~$\ihat$ by a new job~$\ihat'$ that is released at~$r_{\ihat'}:= a_{\ihat} + p_{c'}$, has the same processing time, i.e., ~$p_{\ihat'}= p_{\ihat}$, and has a deadline~$d_{\ihat'}:= \max \{d_{\ihat}, r_{\ihat'}+ (1+\eps)p_{\ihat'} \}$. 
 	
 	In the third step of our modification, we replace every job~$k$ with~$r_k \in [a_{\ihat},r_{\ihat}]$ and~$p_k \leq p_{\ihat}$ by a new job~$k'$ that is released slightly after~${\ihat'}$, i.e.,~$r_{k'}:= r_{\ihat'}+\varrho$ for~$\varrho>0$ sufficiently small. It is important to note that we do not change the processing time or the deadline of~$k'$, i.e.,~$p_{k'}=p_k$ and~$d_{k'} =d_k$. This ensures that~$k'$ finishes on time if and only if~$k$ finishes on time. This modification is feasible, i.e.,~$d_{k'} - r_{k'}\geq(1+\eps)p_{k'}$, because of two reasons. First, \[b_{\ihat} - r_{\ihat'} = b_{\ihat} - (a_{\ihat} + p_{c'}) = b_{\ihat} - a_{\ihat} - (p_{c_1} + p_{c_2} + p_{c_3}) \geq p_{\ihat}\] as~$c_1,c_2,$ and~$c_3$ postponed the region of~$\ihat$ by their processing times in $\I$. Second,~$d_k - b_{\ihat} \geq (1+\frac{\eps}{2}) p_k$ because we only consider jobs that were admitted at some point later than~$b_{\ihat}$ by the region algorithm. Then, 
 	\begin{equation*}	
 	d_{k'} - r_{k'} = d_k - b_{\ihat} + b_{\ihat} - r_{\ihat'}-\varrho \geq (1+\frac{\eps}{2}) p_k + p_{\ihat} - \varrho \geq  (2+\frac{\eps}{2}) p_k -\varrho \geq (1 + \eps) p_{k'},\end{equation*}
 	where the last but one inequality follows from the fact that only jobs with~$p_k \leq p_{\ihat}$ were affected by the modification and the last inequality is due to the sufficiently small choice of~$\varrho$. 
 	
 	So far, we have already seen that the resulting instance is still feasible. It is left to show that~$c'$ completes at~$r_{c'} + p_{c'}$ as well as that~$\ihat'$ is admitted at~$r_{\ihat'}$ and its region ends at~$b_{\ihat}$.  
 	
 	Since it holds that~$p_{c'} < \frac{3\eps}{4} p_{\ihat} < p_{\ihat}$, at~$a_{\ihat}=r_{c'}$ the new job~$c'$ is the smallest available job and any job that was interrupted by~${\ihat}$ is preempted by~$c'$ as well. The jobs in~$T_{-\ihat}$ are released one after the other by Assumption 4 and~$r_{c_1} > a_{\ihat}$. Thus, if~$\ihat'$ has at least one child left after the modification, it holds that~$r_{c_4} = r_{c_1}  + p_{c_1} + p_{c_2} + p_{c_3} = a_{\ihat} + p_{c'} + (r_{c_1} - a_{\ihat})  > r_{\ihat'}$. Hence, no remaining child is released in~$[r_{c'},r_{\ihat'}]$ in the modified instance. Any other job~$k\in T_j$ released in~$[r_{c'},r_{\ihat'}]$ satisfies $p_{k} \geq \frac{\eps}{4} p_{\ihat}$ as $k\notin T_{-\ihat}$. Because~$p_{c'} < p_{\ihat}$, this implies that $p_{k} \geq \frac{\eps}{4} p_{c'}$ holds as well, i.e., no such job~$k$ interrupts~$c'$. Therefore,~$c'$ completes at~$r_i'$.
 	
 	Job~$\ihat'$ is admitted at~$r_{\ihat'}$ if it is the smallest available job at that time. 
 	We have already proven that none of the remaining children of~$\ihat$ is released in~$[a_{\ihat},r_{\ihat'}]$ that might prevent the \region from admitting~$\ihat$ at~$r_{\ihat'}$. Furthermore, the third step of our modification guarantees that any job~$k \in T_j\setminus T_{\ihat}$ that is smaller than~$p_{\ihat}$ is released after~$r_{\ihat'}$. Therefore,~${\ihat'}$ is the smallest available job at time~$r_{\ihat'}$ by construction, and it is admitted. 
	As argued above, the modified instance is still feasible and the interval~$[a_{\pi(\ihat)}, b_{\pi(\ihat)})$ is still completely covered by regions of jobs in~$T_{\pi(\ihat)}$. 
	
	However, the second step of our modification might lead to~$b_{\ihat'} \leq d_{\ihat'}$ which implies that~$\ihat'$ finishes on time while~$\ihat$ does not finish on time. This changes the values of~$|F_j|$ and~$|U_j|$. Clearly, in the case that~$\ihat'$ completes before~$d_{\ihat'}$,~$|U_j'| = |U_j| -1$. By a careful analysis, we see that in this case the number of finished jobs decreases by one as well because the three (on time) jobs~$c_1,c_2$ and~$c_3$ are replaced by only one job that finishes before its deadline. Formally, we charge the completion of~$c'$ to~$c_1$, and the completion of~$\ihat'$ to~$c_2$ which leaves~$c_3$ to account for the decreasing number of finished jobs. Hence, $|F_j'| - |U_j'| = |F_j| - |U_j|$. If~$\ihat'$ does not finish by~$d_{\ihat'}$, then $|F_j'| - |U_j'| = (|F_j| - 2) - |U_j|$. Thus, the modified instance~$\I'$ also satisfies $|F_j'| - |U_j'| < \lfloor \frac{4\ell}{\eps}\rfloor$ but uses less jobs than~$\I$ does. This is a contradiction. 
\end{proof}

%
%

\section{Proofs of Section \ref{sec:comp-ratio} (Competitiveness)}



\subsection{Proof of Lemma \ref{lem:SoL}}

\lemStringOfLeaves*

\begin{proof}
	Observe that $[a_{f}, b_{g}) = \bigcup_{j=1}^k R(j)$ because the leaves $f, \ldots, g$ form a string of jobs. Thus, by showing that there is a job $x \in  X_f^g := \bigcup_{j=f}^g X_{j}$ that satisfies \eqref{eq:VolumeCondition}, we prove the statement with the Volume Lemma. To this end, we show that for every job $f \leq j \leq g$ there exists a set~$Y_j$ such that the processing volume within $Y_j$ is sufficient to cover the interval $[a_{j},b_{j})$ at least $\epsfrac$ times. More precisely, $Y_f,\ldots,Y_g$ will satisfy  
		\begin{enumerate}[label=(\roman*)]
			\item $\bigcup_{j=f}^g Y_j \subset X_f^g$, 
			\item $|Y_j| = \lambda$, and 
			\item $Y_j \subset\{x \in X_f^g: p_x \geq  \beta p_{j}\}$ for every $f\leq j \leq g$. 
		\end{enumerate}
	Indeed, by our choice of $\lambda$, $\sum_{y \in Y_j} p_y \geq \lambda \beta p_{j} = \epsfrac (b_{j} - a_{j})$ if \ref{enum:proofSoL:YSubsetX} to \ref{enum:proofSoL:YBigJobs} are satisfied. Thus, if we choose~$x$ among those jobs in~$X_f^g$ that \opt completes last and guarantee that $x \notin \bigcup_{j=f}^g Y_j$, the Volume Condition \eqref{eq:VolumeCondition} is satisfied. We first give the procedure to find the $Y_f,\ldots,Y_g$ before we show that the constructed sets satisfy \ref{enum:proofSoL:YSubsetX} to \ref{enum:proofSoL:YBigJobs}. 
	
	By assumption, $|X_{f}| > \lambda$. Let $X_{f} = \{x_1,\ldots,x_\lambda,x_{\lambda+1},\ldots\}$ be indexed in increasing completion times $C_x^*$. Define $Y_f :=  \{x_1,\ldots,x_\lambda\}$ and $L_f :=  \{x_{\lambda+1},\ldots\}= X_f \setminus Y_f $, i.e., $Y_f$ contains the $\lambda$ jobs in $X_{f}$ that \opt completes first and $L_f$ contains the last jobs. For $f < j+1\leq g$, let $Y_f,\ldots,Y_j$ and $L_j$ be defined. By assumption, $|X_{j+1 } \cup L_j| > \lambda$ since $|Y_i| = \lambda$ for $1\leq i \leq j$. The jobs in $X_{j +1} \cup L_j = \{x_1, \ldots, x_\lambda, x_{\lambda+1}, \ldots\}$ are again indexed in increasing order of optimal completion times. Then, $Y_{j+1} := \{x_1,\ldots,x_\lambda\}$ and $L_{j+1} :=  \{x_{\lambda+1},\ldots\}$. Since we move jobs only horizontally to later siblings, we call this procedure \pushforward.
	
	By definition, \ref{enum:proofSoL:YSubsetX} and \ref{enum:proofSoL:Y=lambda} are satisfied. Since $f,\ldots,g$ are leaves, the jobs in $Y_j \cap X_j$ are big w.r.t. $j$. Thus, it remains to show that the jobs in $L_j$ are big w.r.t. the next jobs $j+1$. 
	
	To this end, we observe the following. Assume that the jobs in $Y_f,\ldots,Y_j$ are big w.r.t. $f,\ldots,j$, respectively. If we find an index $f\leq \ihat(x) \leq j$ such that $x$ as well as the jobs in $\bigcup_{i=\ihat(x)}^j Y_i$ are released after $a_{\ihat(x)}$, i.e., 
	\begin{equation}\label{eq:proofSoL:release}
	a_{\ihat(x)} \leq r_y \text{ for $y=x$ or $y \in \bigcup_{i = \ihat(x)}^j Y_i$},
	\end{equation}
	and $x$ completes after every $y \in \bigcup_{i = \ihat(x)}^j Y_i$, i.e., 
	\begin{equation}\label{eq:proofSoL:completion}
	C_y^* \leq C_x^*  \text{ for $y \in \bigcup_{i = \ihat(x)}^j Y_i$},
	\end{equation}
	then we can apply the Volume Lemma to show that $x \in L_j$ is big w.r.t. $j+1$.  Indeed, then 	
	\begin{equation*}
	\sum_{i = \ihat(x)}^j \sum_{\substack{y\in X_{i}:  C_y^* \leq C_x^*}} p_y 		
	\geq p_x + \sum_{i = \ihat(x)}^j \sum_{y\in Y_i} p_y 	\geq p_x +  \sum_{i = \ihat(x)}^j \epsfrac (b_i - a_i) =  \epsfrac (b_{j} - a_{\ihat(x)}) + p_x .
	\end{equation*} 	
	We show by induction that such an index $\ihat(x)$ exists for every $x \in L_j$. 
	
	Since $Y_f \subset X_f$, we set $\ihat(x) := f$ for $x\in L_f$. By definition of $L_f$, $C_y^* \leq C_x^*$ for $y \in Y_f$ and $x \in L_f$. Hence, applying the Volume Lemma as explained above shows $p_x \geq \beta p_{f+1}$. 
	
	Let $f < j < g$. Assume that $Y_f,\ldots,Y_j$ and $L_j$ are defined as described above. For jobs $x \in L_j \setminus X_j \subset L_{j-1}$, we have $\ihat(x)$ with the Properties \eqref{eq:proofSoL:release} and \eqref{eq:proofSoL:completion} by induction. For $x \in L_j \cap X_j$, we temporarily set $\ihat(x) := j$ for simplification. We have to distinguish two cases: $\ihat(x)$ also satisfies \eqref{eq:proofSoL:release} and \eqref{eq:proofSoL:completion} for $j$ or we have to adjust $\ihat(x)$. Fix $x\in L_j$. 
	
	\begin{itemize}
		\item $L_i \cap Y_j = \emptyset$ for every $f\leq i < \ihat(x)$. Since only jobs in $L_i$ are shifted to some later job $j$, this implies $\bigcup_{i=f}^{\ihat(x)-1} X_i \cap Y_j = \emptyset$. Thus, the jobs in $Y_j$ are released after $a_{\ihat(x)}$ and by definition, $C_y^* \leq C_x^*$ for $y \in Y_j$. By induction, $x$ and the jobs in $Y_{\ihat(x)} \cup \ldots \cup Y_{j-1}$ satisfy \eqref{eq:proofSoL:release} and \eqref{eq:proofSoL:completion}. Hence, $\ihat(x)$ is a suitable choice for $x$ and $j$. 		
		\item $L_i \cap Y_j \neq \emptyset$ for some $f\leq i < \ihat(x)$. Choose the job $y \in L_{j-1} \cap Y_j$ with the smallest $\ihat(y)$. By a similar argumentation as before, $\bigcup_{i=f}^{\ihat(y)-1} X_i \cap Y_j = \emptyset$, which implies  \eqref{eq:proofSoL:release} for $z \in Y_j$. Again by induction, $y$ and the jobs in $Y_{\ihat(y)} \cup \ldots \cup Y_{j-1}$ satisfy \eqref{eq:proofSoL:release} and \eqref{eq:proofSoL:completion}. Since $x\in L_j$, $C_x^* \geq C_z^*$ for all $z \in Y_j$. This implies $C_x^* \geq C_z^*$ for $z \in \bigcup_{i = \ihat(y)}^{j-1} Y_i$ because $y \in L_{j-1}\cap Y_j$. Set $\ihat(x) := \ihat(y)$.
	\end{itemize}
	As explained above, the Volume Lemma implies $p_x \geq \beta p_{j+1}$. 
	
	The same argumentation holds for $j = g$ although in this special case, Corollary \ref{cor:IsolatedJobs} implies the statement. 
\end{proof}

\subsection{Proof of Lemma \ref{lem:SoN}}\label{subsec:ProofofLem:SoN}

\lemSoN*

\begin{proof}
	We show that for every $j\in J\cup \{M\}$, there exists a partition $(Y_k)_{k\in T_{-j}}$ with 
	\begin{center}
		\begin{enumerate*}[label=(\roman*)]
			\item\label{app:enum:proofSoN:union} $\bigcup_{k\in T_{-j}} Y_k = X_j^S$,
			\item\label{app:enum:proofSoN:big} $Y_k \subset \{ x \in X_j: p_x \geq \beta p_k \}$, and
			\item\label{app:enum:proofSoN:size} $|Y_k| \leq \lambda$ for every $k\in T_{-j}$. 
		\end{enumerate*} 		
	\end{center}
	Then, it holds that $|X_j^S | = |\bigcup_{k\in T_{-j}} Y_k | = \sum_{k \in T_{-j}} |Y_k| \leq \tau_j \lambda$ and, thus, the lemma follows.
	
	The proof consists of an outer and an inner induction. The outer induction is on the distance~$\dist(j)$ of a job~$j$ from machine job~$M$, i.e., $\dist(M) := 0$ and $\dist(j) := \dist(\pi(j)) + 1$ for $j\in J$. The inner induction uses the idea about pushing jobs $x\in X_j$ to some later sibling of $j$ in the same string of jobs (see proof of Lemma~\ref{lem:SoL}).
	
	Let $j \in J$ with $\dist(j) = \dist_{\max} -1 := \max\{\dist(i): i \in J\} -1$. By Observation \ref{obs:XandXj}, $X_j^S = \bigcup_{k: \pi(k) = j} X_k$, where all $k \in T_{-j}$ are leaves at maximal distance from~$M$. We distinguish three cases for $k\in T_{-j}$:
	
	\begin{description}[labelindent=0em ,labelwidth=0cm, labelsep*=1em, leftmargin =\parindent , itemindent = 0pt, style = sameline]
		\item[Case 1.] If $k\in T_{-j}$ is isolated, $|X_k| \leq \lambda$ follows directly from the Volume Lemma as otherwise $\sum_{x \in X_k} p_x \geq \lambda \beta p_k + p_x = \tfrac{\eps}{\eps - \delta} (b_k - a_k) + p_x$ contradicts Corollary~\ref{cor:IsolatedJobs}, where~$x\in X_k$ is the last job that \opt completes from the set~$X_k$. Since all jobs in $X_k$ are big w.r.t.~$k$, we set $Y_k := X_k$.  
		\item[Case 2.] If $k\in T_{-j}$ with $|X_k| > \lambda$ is part of a string, let $f ,\ldots, g$ be the {\em maximal} string satisfying Lemma \ref{lem:SoL} with $k\in \{f,\ldots,g\}.$ With this lemma, we find $Y_f,\ldots,Y_g$ and set $Y_{g+1} := X_{g+1} \cup L_g$. 
		\item[Case 3.] We have not yet considered jobs~$k$ in a string with $|X_k| \leq \lambda$ that 
		do not have siblings~$f,\ldots,g$ in the same string 
		with~$b_g = a_k$ and $\sum_{i=f}^{g} |X_j| > (g - f) \lambda$. This means that such jobs do not receive jobs~$x \in X_i$ for~$i\neq k$ by the \pushforward procedure in Case~2. For such~$k\in T_{-j}$ we define~$Y_k := X_k$. 
	\end{description}	
	Then, $X_j^S = \bigcup_{k: \pi(k) = j} X_k = \bigcup_{k \in T_{-j}} X_k =  \bigcup_{k \in T_{-j}} Y_k$ and, thus, \ref{app:enum:proofSoN:union} to \ref{app:enum:proofSoN:size} are satisfied. 
	
	Let $\dist < \dist_{\max}$ such that $(Y_k)_{k\in T_{-j}}$ satisfying \ref{app:enum:proofSoN:union} to \ref{app:enum:proofSoN:size} exists for all $j\in J$ with $\dist(j) \ge \dist$. Fix~$j \in J$ with $\dist(j) = \dist -1$. By induction and Observation \ref{obs:XandXj}, it holds that $X_j^S =  \bigcup_{k: \pi(k) = j} \left(X_k^B \cup \bigcup_{i \in T_{-k}} Y_i \right) $.
	Now, we use the partitions $(Y_i)_{i \in T_{-k}}$ for $k$ with $\pi(k) = j$ as starting point to find the partition $(Y_k)_{k \in T_{-j}}$. 
	Fix $k$ with $\pi(k)= j$ and distinguish again the same three cases as before.
	
	\begin{description}[labelindent=0em ,labelwidth=0cm, labelsep*=1em, leftmargin = \parindent , itemindent = 0pt, style = sameline]
		\item[Case 1.] 	If $k$ is isolated, we show that $|X_k| \leq \lambda (\tau_k+1) $ and develop a procedure to find $(Y_i)_{i \in T_k}$. Assume for sake of contradiction that $|X_k| > \lambda (\tau_k+1) $ and index the jobs in~$X_k$ in increasing order of completion times, i.e., $X_k = \{x_1, \ldots, x_{\lambda(\tau_k + 1) }, x_{\lambda(\tau_k + 1) +1 },\ldots \}$, and set $L :=  \{ x_{\lambda(\tau_k + 1) +1 },\ldots \}$. Then, \begin{equation*}
		|X_k^B\setminus L | = |X_k\setminus L | - |X_k^S\setminus L | = (\tau_k+1) \lambda - \sum_{i \in T_{-k}} |Y_i\setminus L | = \lambda + \sum_{i \in T_{-k}} (\lambda - |Y_i\setminus L |).
		\end{equation*}
		By induction hypothesis, $\lambda - |Y_i\setminus L | \geq 0$ for~$i \in T_{-k}$. Let $Y_k$ contain $\lambda$ arbitrary big jobs in $X_k^B\setminus L$ and assign each $Y_i$ for $i\in T_{-k}$ exactly $\lambda - |Y_i\setminus L |$ of the remaining (big) jobs in $X_k^B\setminus L$. This is possible because the jobs in  $X_k^B$ are big for any descendant of $k$, i.e., they satisfy \ref{enum:proofSoN:big}. By choice of $\lambda$, each of the just obtained sets covers the region of the corresponding job at least $\epsfrac$ times. Thus, the jobs in $X_k \setminus L$ have a total processing volume of at least $\epsfrac(b_k - a_k)$. Therefore, any job $x \in L$ satisfies \eqref{eq:VolumeCondition} which contradicts the fact that $k$ is isolated by Corollary \ref{cor:IsolatedJobs}. Thus,~$|X_k| \leq \lambda(\tau_k +1)$. 		

		To construct $(Y_i)_{i \in T_{k} }$, we assign $\min \{\lambda, |X_k^B|\}$ jobs from $X_k^B$ to $Y_k$. If $|X_k^B| > \lambda$, distribute the remaining jobs according to $\lambda - |Y_i|$ among the descendants of $k$. Then, $X_k = \bigcup_{i \in T_{k}} Y_i$. Because a job that is big w.r.t job $k$ is also big w.r.t.\ all descendants of $k$, 
		every (new) set $Y_i$ satisfies \ref{enum:proofSoN:big} and \ref{enum:proofSoN:size}. We refer to this procedure as \pushdown since jobs are shifted vertically to descendants.
		
		\item[Case 2.] 	If $|X_k| > \lambda (\tau_k+1) $, $k$ must belong to a string with similar properties as described in Lemma~\ref{lem:SoL}, i.e., there are jobs $f,\ldots,g$ containing $k$ such that 
		\begin{enumerate}
			\item $\sum_{j=f}^i |X_{j}| > \lambda\sum_{j=f}^i \tau_j $ for all $f\leq i\leq g$ and
			\item\label{enum:lem:SoN} $b_{j} = a_{j+1}$ for all $f\leq j < g$.
		\end{enumerate}
		Choose $\{f,\ldots,g\}$ maximal with those two properties. We show that the Volume Lemma implies the existence of another sibling $g+1$ that balances the sets $X_f,\ldots,X_g,X_{g+1}$. This is done by using the \pushdown procedure within a generalization of the \pushforward procedure. 
		
		As the jobs $f,\ldots,g$ may have descendants, we use \pushforward to construct the sets $Z_f,\ldots,Z_g$ and $L_f,\ldots, L_g$ with $|Z_k| = \lambda (\tau_k+1)$. Then, we show that we can apply \pushdown to $Z_k$ and $(Y_i)_{i\in T_{-k}}$ in order to obtain $(Y_i)_{i \in T_{k} }$. 
		This means the newly obtained partition satisfies 
		\begin{enumerate}[label=(\roman*)]
			\setcounter{enumi}{3}
			\item\label{enum:proofSoN:union'} $Y_k\cup\bigcup_{i\in T_{-k}} Y_i = Z_k$,
			\item $Y_i \subset \{ x \in X_j: p_x \geq \beta p_i \}$ and
			\item\label{enum:proofSoN:size'} $|Y_i| = \lambda$ for every $i\in T_{k}$. 
		\end{enumerate} 
		This implies that the set $Z_k$ covers $[a_k,b_k)$ at least $\epsfrac$ times. Thus, the sets~$X_k$ with~$f\leq k\leq g$ satisfy~\eqref{eq:VolumeCondition} and we can apply~Corollary~\ref{cor:IsolatedJobs}.		
		
		To define $Z_f,\ldots,Z_g$, we index the jobs in $X_f = \{x_1, \ldots, x_{\lambda_f}, x_{\lambda_f + 1},\ldots\}$ in increasing order of optimal completion times and set $Z_f := \{x_1, \ldots, x_{\lambda_k}\}$ and $L_f = X_f\setminus Z_f$. Assume that $Z_f,\ldots,Z_k$ and $L_f,\ldots,L_k$ are defined. Index the jobs in $X_{k+1} \cup L_k = \{x_1, \ldots, x_{\lambda_{k+1}}, x_{\lambda_{k+1} + 1},\ldots\}$ in increasing order of completion times and set $Z_{k+1} := \{x_1, \ldots, x_{\lambda_{k+1}}\}$ and $L_{k+1} = (X_{k+1} \cup L_k)\setminus Z_{k+1}$. Use the \pushdown procedure to obtain the partition $(Y_i)_{i \in T_{k}}$.
		
		If we can show that any job $x \in L_k$ is big w.r.t. $k+1$, we have that $Z_{k+1} \setminus X_{k+1}^S$ only contains big jobs w.r.t. ${k+1}$, which are also big w.r.t. every $ i \in T_{-(k+1)}$. As in Case 1,
		\begin{equation*}\label{eq:NumberOfBigJinZ}
		|Z_{k+1} \setminus X_{k+1}^S| = |Z_{k+1}| - |X_{k+1}^S \setminus L_{k+1}| = \lambda + \sum_{i \in T_{-(k+1)}} (\lambda - |Y_i\setminus L_{k+1}|).
		\end{equation*}
		Hence, the just defined partition $(Y_i)_{i \in T_{k}}$ satisfies \ref{enum:proofSoN:union'} to \ref{enum:proofSoN:size'}. 
		
		As in the proof for Lemma \ref{lem:SoL}, we show by induction that every $x \in L_k$ exhibits an index $\ihat(x)$ with 
		\begin{align}
		a_{\ihat(x)} & \leq r_y  \label{eq:proofSoN:release} \\
		C_y^* &\leq C_x^* \label{eq:proofSoN:completion}
		\end{align}
		for $y = x$ or $y \in \bigcup_{i = \ihat(x)}^j Z_i$. Then, the Volume Lemma implies that $p_x \geq \beta p_{k+1}$.
		
		For $x\in L_f$, set $\ihat(x) = f$. Thus, Equations \eqref{eq:proofSoN:release} and \eqref{eq:proofSoN:completion} are trivially satisfied. Since $Z_f \subset X_f$,  we have that $Z_f \setminus X_f^S$ only contains big jobs w.r.t. $f$.  
		
		Let $f < k < g$. Assume that $Z_f,\ldots,Z_k$ and $L_k$ are defined as described above. For jobs $x \in L_k \setminus X_k$, we have $\ihat(x)$ with the Properties \eqref{eq:proofSoN:release} and \eqref{eq:proofSoN:completion} by induction. For $x \in L_k \cap X_k$, we temporarily set $\ihat(x) := k$ for simplification. We have to distinguish two cases: $\ihat(x)$ also satisfies \eqref{eq:proofSoN:release} and \eqref{eq:proofSoN:completion} for~$k$ or we have to adjust $\ihat(x)$. Fix $x\in L_k$. 		
		\begin{itemize}
			\item $L_i \cap Z_k = \emptyset$ for every $f\leq i < \ihat(x)$. Since only jobs in $L_i$ are shifted to some later job $k$, this implies $\bigcup_{i=f}^{\ihat(x)-1} X_i \cap Z_k = \emptyset$. Thus, the jobs in $Z_k$ are released after $a_{\ihat(x)}$ and by definition, $C_y^* \leq C_x^*$ for $y \in Z_k$. By induction, $x$ and the jobs in $Z_{\ihat(x)} \cup \ldots \cup Z_{k-1}$ satisfy \eqref{eq:proofSoN:release} and \eqref{eq:proofSoN:completion}. Hence, $\ihat(x)$ is a suitable choice for $x$ and $k$. 
			
			\item $L_i \cap Z_k \neq \emptyset$ for $f\leq i < \ihat(x)$. Choose the job $y \in L_{k-1} \cap Z_k$ with the smallest $\ihat(y)$. By a similar argumentation as before, $\bigcup_{i=f}^{\ihat(y)-1} X_i \cap Z_k = \emptyset$, which implies  \eqref{eq:proofSoN:release} for $z \in Z_k$. Again by induction, $y$ and the jobs in $Z_{\ihat(y)} \cup \ldots \cup Z_{k-1}$ satisfy \eqref{eq:proofSoN:release} and \eqref{eq:proofSoN:completion}. Since $x\in L_k$, $C_x^* \geq C_z^*$ for all $z \in Z_k$. This implies $C_x^* \geq C_z^*$ for $z \in \bigcup_{i = \ihat(y)}^{k-1} Z_i$ because $y \in L_{k-1}\cap Y_k$. Set $\ihat(x) := \ihat(y)$.
		\end{itemize}
		As explained above, the Volume Lemma implies $p_x \geq \beta p_{k+1}$. 
		
		For $k+1 = g$, the above argumentation can be combined with Corollary \ref{cor:IsolatedJobs} to prove that the sibling $g+1$ indeed exists. Set $Z_{g+1} := X_{g+1} \cup L_g$ and use \pushdown to construct $(Y_{i})_{i \in T_{(g+1)}}$.

		\item[Case 3.] Any job $k$ with $\pi(k) = j$ that is part of a string and was not yet considered must satisfy $|X_{k} | \leq (\tau_k+1) \lambda$. 
		We use the \pushdown procedure for isolated jobs to get the partition~$(Y_{i})_{i \in T_k }$.
	\end{description}
	Hence, we have found $(Y_k)_{k\in T_{-j}}$ with the properties \ref{enum:proofSoN:union'} to \ref{enum:proofSoN:size'}.	
\end{proof}

\subsection{Analysis of the Region Algorithm is Tight}\label{sec:TightAnalysis}

The region algorithm is best possible (up to constants) for scheduling without commitment as the matching lower bound in Theorem~\ref{thm:detLB-w=1-anytime} proves. Consider now algorithms that must commit to job completions. We show that the analysis of the region algorithm is tight in the sense that the competitive ratio of the region algorithm is $\Omega(\alpha/\beta)$. Moreover, we give examples that show that for the commitment upon admission model the choice~$\alpha \in \Omega(1/\eps)$ and~$\beta \in O(1/\eps)$ is best possible. 

\begin{lemma}\label{lem:GeneralLBOnCompRatio}
	Let~$0 < \eps \leq 1 $, $\alpha\geq 1$, and~$0 < \beta < 1$. Then, the competitive ratio of the region algorithm is bounded from below by~$\alpha/\beta$.
\end{lemma}

\begin{proof}
	We consider an instance where a job~$0$ with processing time~$p_0 =1$ and a huge scheduling interval~$[r_0, r_0+\alpha+2)$ is released first. Then, the region algorithm blocks the region~$[r_0, r_0 +\alpha]$ for this job. During this interval, $\lfloor \alpha / \beta \rfloor$ jobs of size~$p_j = \beta$ arrive. They all fit into~$R(0)$ but the jobs are to big relative to~0 to be admitted. Then, an offline optimum would process all small jobs until $r_0 + \alpha$ before starting job~$0$. Hence, the \region completes one job while it is optimal to complete $\lfloor \alpha / \beta \rfloor + 1$ jobs. 
	
	More formally, let $r_0 = 0$, $p_0 = 1$ and~$d_0 = \alpha + 1$. Fix $0 < \varphi < \beta < 1$. For $1 \leq j \leq \lfloor \alpha/\beta \rfloor$ let~$r_j = (j-1)\beta + \varphi$, $p_j = \beta$ and~$d_j = r_j + (1+\eps) p_j$. The \region admits job~$0$ at time~$0$ and blocks the interval~$[0,\alpha)$ for~$0$. Thus, the \region cannot admit any of the small jobs and completes only job~$0$. This behavior does not depend on the commitment model.
	
	An optimal offline algorithm processes the jobs $1,\ldots,\lfloor \alpha/\beta \rfloor$ one after the other in the interval~$[\varphi, \lfloor \alpha / \beta \rfloor \beta + \varphi ) \subset [0, \alpha + 1)$. At the latest at time~$\alpha+1$ job~0 starts processing and finishes on time. 
	
	Thus, the competitive ratio of the algorithm is bounded from below by $\lfloor \alpha / \beta \rfloor + 1 \geq \alpha / \beta$. 
\end{proof}

\begin{lemma}\label{lem:CommitmentLBOnCompRatio}
	The competitive ratio of the region algorithm in the scheduling with commitment model is bounded from below by $\Omega(1/\eps^2)$. 
\end{lemma}

\begin{proof}
	The proof consists of two parts. First we show an upper bound on the choice of~$\beta$ in terms of~$\delta$. Then, we use this observation to show an upper bound on $\beta$ depending on~$\alpha$. 
 	
 	It is obvious that $\beta \leq \delta$ must hold as otherwise a job that is admitted at~$d_j - (1+\delta)p_j$ and interrupted by another job~$i$ with $p_i=\beta p_j -\varphi$ cannot finish on time. Hence, $\beta < \delta \leq 1$ must hold. 
 	
 	We define a family of instances~$\I_m(c)$ that depends on two natural numbers~$m, c \in \N$ where $c$ is chosen such that 
 	\begin{equation}\label{eq:alpha,beta,c}
 		\frac{1}{\beta(c+1) } < \alpha \leq \frac{1}{\beta c}.
 	\end{equation} Each instance consists of four types of jobs, a job $0$ that cannot be finished unless~$\alpha$ and~$\beta$ satisfy certain bounds, an auxiliary job~$-1$ that guarantees that~$0$ is not admitted before~$d_0 - (1+\delta) p_0$ and two sets of jobs, $B(c)$ and $G(n)$ that block as much time in~$[a_0,d_0)$ as possible. A visualization of the instance can be seen in Figure \ref{fig:TightAnalysis}.
 	
 	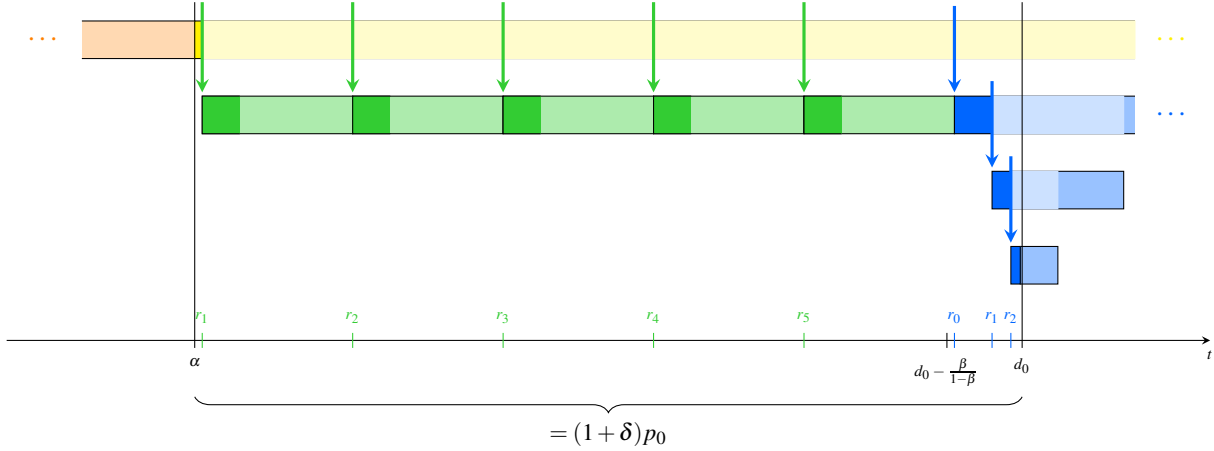
\begin{figure}[tb]
 		\begin{center}
 			\begin{tikzpicture} 
 			
 			\def \alfa {2}
 			\def \d {13}
 			\def \axis {-2}
 			
 			\draw[-stealth] (-.5,\axis) to (15.5,\axis) node [below,font=\tiny] {$t$};
 			\draw (\alfa, 2.5) to (\alfa, \axis-.1) node[below, font=\tiny] {$\alpha$};
 			
 			\draw (\d-1, \axis+.1) to (\d-1, \axis-.1) node [below,font=\tiny] {$d_0 - \frac{\beta}{1-\beta}$};
 			\draw [decorate,decoration={brace,amplitude=6pt}] (\d, \axis-.75) -- (\alfa,\axis-.75)  node [black,midway,yshift = -.5cm,font=\footnotesize] {$= (1+\delta) p_0$} ;

 			\filldraw[fill=orange!30, draw=black] (.5,1.75) -- (2,1.75) -- (2,2.25) -- (.5,2.25);
 			\node[orange] at (0,2) {$\cdots$};
 			
 			\filldraw[fill=yellow, draw=black]  (14.5,2.25) -- (2,2.25) -- (2,1.75) -- (14.5,1.75);
 			\node[yellow] at (15,2) {$\cdots$};
 			
 			\fill[fill=yellow!20, draw=yellow!20] (\alfa+.1, 1.75) rectangle (\d +1.5,2.25);
 			
 			\foreach \job in {1,...,5}{
 				\draw [-stealth, green, very thick] (\alfa+.1+2*\job-2,2.5) to (\alfa+.1+2*\job-2,1.3);
 				\draw [green] (\alfa+.1+2*\job-2, \axis-.1) to (\alfa+.1+2*\job-2, \axis+.1) node [above,font=\tiny] {$r_{\job}$};		
 				\filldraw[fill = green!40, draw=black] (\alfa+.1+2*\job-2, .75) rectangle (\alfa+.1+2*\job, 1.25);
 				\fill [fill = green, draw=black] (\alfa+.1+2*\job-1.5, 1.25) -- (\alfa+.1+2*\job-2 , 1.25) --  (\alfa+.1+2*\job-2, .75) -- (\alfa+.1+2*\job-1.5, .75);	
 			}
 			
 			\def \m {2}
 			
 			\def \factor {.5}
 			\def \size {{.5,.5^2,.5^3}}
 			
 			\foreach \job in {0,...,2}{
 				\edef \position {0}
 				\edef \length {0}
 				\foreach \i in {0,...,\job}{
 					\xdef \position {\position + \size[\i]}
 				}
 				\foreach \i in {\job,...,\m}{
 					\xdef \length {\length + 4*\size[\i]}
 				}

 				\filldraw[fill=blue!40, draw=black] (\d-.9 + \position,  .75 - \job) rectangle (\d-.9 + \position + 4*\length, 1.25 - \job);
 				\if\job0
 				\else 
 				\filldraw[fill=blue!20, draw=blue!20] (\d-.9 + \position - \size[\job],  1.75 - \job ) rectangle (\d-.9 + \position + 4*\length, 2.25 - \job);
 				\fi
 			}
 			
 			\foreach \job in {0,...,2}{
 				\edef \release {0}
 				\foreach \i in {0,...,\job} {
 					\xdef \release {\release + \size[\i]}
 				}
 				\fill [fill = blue, draw=black] (\d-.9+\release-\size[\job], .75-\job) rectangle (\d-.9+\release, 1.25-\job);
 				\draw [-stealth, blue, very thick] (\d-.9+\release-\size[\job], 1.5-\job + .95) to (\d-.9+\release-\size[\job],.75-\job+.55);
 				\draw [blue] (\d-.9+\release-\size[\job], \axis-.1) to (\d-.9+\release-\size[\job], \axis+.1) node [above,font=\tiny] {$r_{\job}$};
 			}

 			\fill [white] (14.5,.5) rectangle (16.5,1.5);
 			\node [blue] at (15,1) {$\cdots$};
 			\draw (\d, 2.5) to (\d, \axis-.1) node [below,font=\tiny] {$d_0$};
 			
 			%
 			%
 			\end{tikzpicture}
 		\end{center}
 		\caption{The structure of the regions and the schedule generated by the \region when faced with the instance~$\I_m(c)$. The darkest shades of a color mean that jobs are scheduled there. The light yellow and blue parts show that the region is currently interrupted. The only time slots where~$0$ can be processed are the lighter parts of the green regions, i.e., the regions belonging to~$B(c)$.} \label{fig:TightAnalysis}
 	\end{figure}
 	
 	More precisely, at time $t=0$, an auxiliary job~$-1$ is released with~$p_{-1} =1$ and $d_{-1} = (1+\eps) p_{-1}$. The \region admits this job and assigns it the region~$R(-1) = [0,\alpha)$. At time~$\alpha-(\eps-\delta)$ job~0 is released with $p_0 =1$ and $d_0 = \alpha + 1 + \delta$. Obviously, this job is admitted at time~$\alpha$ as it is still available. Fix~$\varphi >0$ sufficiently small.
 	
 	At time $\alpha + \varphi$ the sequence~$B(c)$ of $c$ identical jobs is released one after the other such that the release date of one job coincides with the end of the region of the previous job. For~$0 \leq i \leq c-1$, a tight job is released at $r_i := \alpha + i/c + \varphi $ with processing time~$p_i = \beta -\varphi$ and deadline~$d_i = r_i + (1+\eps)p_i$. Since \[ r_i + \alpha p_i \leq \alpha + i/c + \varphi + \beta / (\beta c) - \alpha \varphi  < \alpha + (i+1)/c + \varphi = r_{i+1} \] each of these jobs is admitted by the region algorithm at their release date. The last of these regions ends at $\alpha + (c-1)/c + \varphi + \alpha (\beta - \varphi) = \alpha + \frac{c-1}{c} + \varphi + \frac{1}{c} - \alpha \varphi \leq \alpha +1$. Thus, in the limit $\varphi \rightarrow 0$, they block $c\beta$ units of time in $[a_j, a_j + d_j)$. 
 	
 	At time $d_0 - \frac{\beta}{1-\beta}$, a sequence of~$m$ geometrically decreasing jobs~$G(m)$ is released. For~$1\leq j \leq m$, job~$j$ is released at~$r_j =d_0 - \frac{\beta}{1-\beta} + \sum_{i =1}^j \beta^i$ with processing time~$p_j = (\beta - \varphi)^j$ and deadline~$d_j = r_j + (1+\eps)p_j$. Then, $p_{j+1} = (\beta - \varphi) p_j < \beta p_j$. Thus, the region algorithm admits each of the~$m$ jobs. Again, in the limit~$n\rightarrow \infty$ and $\varphi \rightarrow 0$, the processing volume of~$G(m)$ sums up to $\frac{\beta}{1-\beta}$. 
 	
 	Putting the two observations together, we obtain \[
 		\frac{\beta}{1-\beta} + c \beta \leq \delta 
 	\] 
 	as otherwise job~$0$ cannot finish on time. Hence, 
		$0  \leq c \beta ^2 - (1 + c + \delta)\beta + \delta.$
	Solving for the two roots, $\beta_+$ and~$\beta_-$, we obtain \begin{align*}
		\beta_{+} & =\ \frac{ 1 + c + \delta + \sqrt{ (1 + c + \delta)^2 - 4 c \delta} }{ 2c } 
		\ =\ \frac{ 1 + c + \delta + \sqrt{(1+c)^2 + 2(1+c)\delta + \delta^2 - 4c\delta   } }{ 2c } \\
		& \geq\ \frac{ 1 + c + \delta + \sqrt{ c^2 + \delta^2 - 2c\delta  } }{ 2c } 
		\ =\ \frac{ 1 + c + \delta + \sqrt{ (c-\delta)^2 } }{ 2c } \\
		& >\ 1.
	\end{align*}
	As we have seen by the first example,~$\beta < 1$ must hold. Thus, we conclude that the only valid choice for~$\beta$ is in the interval~$(0,\beta_-)$. By a similar calculation, it follows that $\beta_{-} \leq \frac{\delta}{c}.$
	As we know by Lemma~\ref{lem:GeneralLBOnCompRatio}, the competitive ratio is bounded from below by $\alpha / \beta$. Combined with the two bounds on~$\alpha$,~$\frac{1}{\beta(c+1)} < \alpha \leq \frac{1}{\beta c}$, we obtain \[
		\frac{\alpha}{\beta} \geq \frac{1}{\beta(c+1) } \frac{1}{\beta } = \frac{c^2}{\delta^2 (c+1)}.
	\] Since the right hand side is increasing in~$c$ for positive~$c$, the expression is minimized for~$c=1$. This implies that~$\beta \in \OO(\eps)$ and therefore~$\alpha \in \Omega(1/\eps)$.
\end{proof}

\section{Additional Lower Bounds}\label{app:lbs}



\subsection{Lower Bounds: Commitment upon arrival}



\subsubsection{Proportional weights ($w_j=p_j$)}\label{sec:LBCommUponReleasewj=pj}

We consider the setting of proportional weights in which $w_j=p_j$ for all jobs $j$. It has been previously known that deterministic algorithms can achieve a competitive ratio of $\Theta(\frac{1}{\eps})$ for this setting~\cite{DBLP:conf/approx/DasGuptaP00}. In Theorem~\ref{thm:wj=pj:admission:LB} below, we provide a lower bound of $\Omega(\log \frac1\eps)$ for randomized algorithms in the less restrictive setting of commitment upon job admission. This implies the following corollary.

\begin{corollary}
	Consider proportional weights ($w_j=p_j$). The competitive ratio of any randomized algorithm is $\Omega(\log1/\eps)$ for scheduling with commitment upon arrival.
\end{corollary}


\subsection{Lower bounds: Commitment on job admission and $\delta$-commitment}

Lower bounds for arbitrary weights and commitment on job admission have already been shown~\cite{DBLP:conf/spaa/LucierMNY13}. We rule out bounded performance ratios even for unit weights and give a lower bound for proportional weights.

\subsubsection{Unit weights ($w_j=1$)}

The lower bound for scheduling without commitment in Theorem~\ref{thm:detLB-w=1-anytime} immediately  carries over to any model for scheduling with commitment. However, we provide a much simpler proof for the setting we study~here.

\begin{proof}[Simple proof of Theorem~\ref{thm:detLB-w=1-anytime} for commitment on job admission and $\delta$-commitment]
	Let $\eps<\frac14$ and $\delta<\eps$. At time $0$, job $1$ is released with $p_1=1$ and $d_j=1+\epsilon$. To be competitive in the case that no further jobs are released, the algorithm has to finish this job. Hence, independently of the respective commitment model, the algorithm needs to have committed to this job at time $2\eps$. At this time, there are $\lfloor\frac{1-\eps}{\eps}\rfloor=\Omega(\frac1\eps)$ jobs with processing time $\eps$ and deadline $1$ released. While the optimum can finish all of them, the laxity of job~$1$ admits only that the online algorithm accepts one additional job.
\end{proof}

\subsubsection{Proportional weights ($w_j=p_j$)}

We first consider deterministic algorithms.

\begin{theorem}\label{thm:wj=pj:admission:LB}
	Consider proportional weights ($w_j=p_j$). For commitment on job admission and the $\delta$-commitment model, the competitive ratio of any deterministic algorithm is $\Omega(\frac1\eps)$.
	\end{theorem}
	\begin{proof}
		Let $\eps< \frac 1 4 $ and $\delta<\eps$. At time $0$, job $1$ is released with $p_1=1$ and $d_j=1+\epsilon$. To be competitive in the case that no further jobs are released, the algorithm has to complete this job. Hence, independently of the respective commitment model, the algorithm needs to have committed to this job at time $2\eps$. At this time, job $2$ is released, where $p_2=\frac{1-3\eps}\eps=\Omega(\frac 1\eps)$ and its slack is $1-3\eps$. The optimum only schedules job $2$, but the online algorithm cannot finish it because it has committed to job $1$, which still requires $1-2\eps>1-3\eps$ processing.
	\end{proof}

We show a weaker lower bound for randomized algorithms. 

\begin{theorem}\label{thm:wj=pj:admission:LB}
	Consider proportional weights ($w_j=p_j$). For commitment on job admission and the $\delta$-commitment model, the competitive ratio of any randomized algorithm is $\Omega(\log \frac 1 \eps)$.
	\end{theorem}

\begin{proof}
	Let $k=\lfloor\log(\frac1{8\eps})\rfloor$, and consider a $c$-competitive algorithm. The adversary releases at most $k$ jobs, where job $j=1,\dots,k$ arrives at $r_j = 2\eps\sum_{i=1}^{j-1}2^{i-1}$, has processing time $2^{j-1}$ and slack $\eps 2^{j-1}$.
	
	Denote by $n_i$ the probability that the algorithm commits to job $i$. We make the following observations:
\begin{compactenum}[(i)]
	\item The release date of job $j$ is $$2\eps\sum_{i=1}^{j-1}2^{i-1}<2\eps\cdot2^{\log(\frac1{8\eps})}\leq \frac14,$$ at which time any job $j'<j$ that the algorithm has committed to has at least $p_1-1/4=3/4$ processing left. The slack of $j$ is however only at most $$\eps\cdot2^{j-1}\leq\eps\cdot2^{\lfloor\log(\frac1{8\eps})\rfloor-1}\leq\frac{1}{16}.$$ This implies that no two jobs can both be committed to at the same time. Hence, $\sum_{i=1}^k n_i\leq 1$.
	\item The algorithm has to commit to $j<k$ at the latest at $$r_j+\eps2^{j-1}=2\eps\sum_{i=1}^{j-1}2^{i-1}+\eps2^{j-1}<2\eps\sum_{i=1}^{j}2^{i-1}=r_{j+1},$$ that is, unknowingly whether $j+1$ will be released or not, so it has to be competitive with the optimum that only schedules $j$. Hence, we have $\sum_{i=1}^j n_i\cdot 2^{i-1}\geq \frac{2^{j-1}}c$.
\end{compactenum}
	This allows us to apply Lemma~\ref{lem:algebraic-fact} to $n_1,\dots,n_k$, showing $c\geq \frac{k+1}2\Omega(\log\frac1\eps)$.
\end{proof}






\newpage
\bibliographystyle{abbrv} 
\bibliography{throughput}
\end{document}